\documentclass[11pt]{amsart}
\usepackage[margin=1in]{geometry}
\usepackage{amsmath,amssymb,amsthm,mathtools}
\usepackage{graphicx,booktabs}
\usepackage[numbers,sort&compress]{natbib}
\usepackage{xcolor}
\usepackage[colorlinks=true,linkcolor=blue!50!black,citecolor=blue!50!black,urlcolor=blue!50!black]{hyperref}
\usepackage{algorithm,algpseudocode}
\usepackage{enumitem}
\newtheorem{theorem}{Theorem}[section]
\newtheorem{proposition}[theorem]{Proposition}
\newtheorem{lemma}[theorem]{Lemma}
\newtheorem{corollary}[theorem]{Corollary}
\theoremstyle{definition}
\newtheorem{assumption}{Assumption}

\newtheorem{example}[theorem]{Example}
\theoremstyle{remark}
\newtheorem{remark}[theorem]{Remark}

\newcommand{\E}{\mathbb{E}}
\newcommand{\R}{\mathbb{R}}
\newcommand{\Prob}{\mathbb{P}}
\newcommand{\F}{\mathcal{F}}
\newcommand{\tr}{^{\top}}
\newcommand{\dd}{\,\mathrm{d}}
\newcommand{\db}{d_\beta}
\newcommand{\dg}{d_g}

\DeclareMathOperator{\diag}{diag}
\DeclareMathOperator{\dist}{dist}
\DeclareMathOperator{\conv}{conv}

\title[Efficient stochastic GMM]{An efficient stochastic generalized method of moments estimator for nonlinear moment models}
\author{D. Nguyen}
\address{Department of Mathematics, Marist University, 3399 North Road, Poughkeepsie, NY 12601, United States}
\email{nducduy@gmail.com}


\begin{document}

\begin{abstract}
We study a recursive (online) estimator for parameters defined by possibly nonlinear and overidentified moment restrictions $\E[g(X,\beta^*)]=0$: the nonlinear stochastic GMM algorithm proposed by Chen, Lee, Liao, Seo, Shin and Song, who did not establish its asymptotic properties. Each observation updates, at a cost that does not depend on the sample size, an estimate of the Jacobian of the moments, the inverse of an estimate of their covariance matrix and a Gauss--Newton preconditioner; for moments of instrumental-variable type, linear or not, we give an exact rank-two Sherman--Morrison--Woodbury update of the preconditioner. The parameter is moved by a preconditioned stochastic approximation step and the iterates are averaged. Under a preconditioned monotonicity condition, which holds globally in linear models and, in general, on a neighbourhood of a consistent pilot estimate, we prove almost sure convergence, almost sure rates and a functional central limit theorem with the efficient GMM variance, assuming four moments of the moment function and its Jacobian when the step exponent exceeds $2/3$. The inverse Gauss--Newton matrix that enters the preconditioner is a strongly consistent estimator of this variance, and a projected linearised statistic gives an online test of overidentifying restrictions. We extend the results to $L$-dependent data with long-run variance weights, under a stronger moment condition. In simulations, the averaged estimator is within a few percent of offline two-step GMM in accuracy at $10^6$ observations and has similar coverage; with up to $10^7$ observations, it uses a fraction of the computing time of offline GMM and a fixed amount of memory; and the online test has correct size and the same power as the offline test. In an application to the effect of a third child on weeks worked, with $254{,}654$ women from the 1980 US Census, single-pass estimates agree with offline GMM within a fraction of a standard error on average over orderings of the data, and the online and offline overidentification tests lead to the same conclusion in most orderings of the data.
\end{abstract}

\maketitle
\tableofcontents

%=====================================================================
\section{Introduction}
%=====================================================================

The generalized method of moments (GMM) of \citet{Hansen1982} estimates a parameter $\beta^*\in\R^{\db}$ from moment restrictions
\begin{equation}\label{eq:moment}
\E\big[g(X,\beta^*)\big]=0,\qquad g:\mathcal X\times\R^{\db}\to\R^{\dg},\quad \dg\ge\db,
\end{equation}
without specifying the distribution of the data $X$. It is a basic tool in econometrics and statistics: instrumental-variable (IV) regression, Euler equations of consumption-based asset pricing models \cite{HansenSingleton1982}, models with endogenous regressors and count data \cite{Mullahy1997,WindmeijerSantosSilva1997} and moment-based estimation of stochastic volatility models \cite{AndersenSorensen1996} are all of the form \eqref{eq:moment}. The efficient two-step GMM estimator minimises $\bar g_n(\beta)\tr W_n\bar g_n(\beta)$, where $\bar g_n$ is the sample average of the moments and $W_n$ estimates the inverse of their covariance matrix \cite{Hansen1982,NeweyMcFadden1994,Hayashi2000}. For nonlinear moments this requires an iterative optimization in which every iteration passes over the whole sample, and the whole sample must be kept in memory. With very large or streaming data sets, this is the bottleneck.

\citet{ChenEtAl2025} proposed a stochastic approximation alternative, called SGMM, which processes the observations one at a time. For linear IV models they established almost sure convergence of the online two-stage least squares (2SLS) estimator, asymptotic normality and functional central limit theorems for online 2SLS and for the efficient SGMM, and online versions of the Durbin--Wu--Hausman and Hansen--Sargan tests. They also described a nonlinear version of SGMM, in which the Jacobian and the weight matrix are updated at the averaged iterate, but \emph{left its asymptotic properties to future research}. The purpose of this paper is to provide this theory.

We consider the nonlinear SGMM of \citet{ChenEtAl2025}, with a step-size offset added for numerical stability, and a projection onto a compact set and a truncation of the preconditioner added for the analysis (Section~\ref{sec:algo}). After an initialization on a small pilot sample, the $i$th observation $X_i$ updates the running average $\Phi_i$ of the Jacobians $\partial_\beta g(X_j,\bar\beta_{j-1})$, the inverse $W_i$ of the running average of $g(X_j,\bar\beta_{j-1})g(X_j,\bar\beta_{j-1})\tr$ by a Sherman--Morrison formula, and the inverse $H_i=(\Phi_i\tr W_i\Phi_i)^{-1}$ of the Gauss--Newton matrix; the parameter is updated by
\[
\beta_i=\beta_{i-1}-\gamma_i H_{i-1}\Phi_{i-1}\tr W_{i-1}\,g(X_i,\beta_{i-1}),
\]
and the estimator is the average $\bar\beta_n$ of the iterates \cite{Ruppert1988,PolyakJuditsky1992}. The step $H_{i-1}\Phi_{i-1}\tr W_{i-1}g$ is a stochastic Gauss--Newton step for the GMM objective, in the spirit of the stochastic Newton and Gauss--Newton algorithms of \citet{BercuGodichonPortier2020} and \citet{CenacGodichonPortier2025}, here driven by an estimating equation rather than by the gradient of a loss.

\paragraph{Our paper makes the following contributions to the literature.}
\begin{enumerate}[leftmargin=*,itemsep=2pt]
\item \emph{Recursive inversion.} For moment functions of IV type, $g(x,\beta)=z\,r(x,\beta)$ with instruments $z$ and a scalar residual $r$, the Jacobian has rank one and the efficient weight is updated with a multiple of $z$; we show that $H_i$ is then obtained from $H_{i-1}$ by an exact rank-two Woodbury update, which costs $O(\dg^2+\dg\db+\db^2)$ operations per observation (Lemma~\ref{lem:woodbury}). It contains the formula of \citet{ChenEtAl2025} for linear 2SLS, for which we give a complete proof.
\item \emph{Almost sure convergence in the nonlinear case.} Under a preconditioned monotonicity condition (Assumption~\ref{ass:mono}), which holds globally for linear models and on every sufficiently small neighbourhood of $\beta^*$ for smooth nonlinear models (Proposition~\ref{prop:local}), the iterates and their average converge almost surely, with a continuously updated weight matrix (Theorem~\ref{thm:as}). With the parameter set chosen as a ball around a pilot estimate, convergence holds with probability tending to one as the pilot sample grows, without knowledge of $\beta^*$ (Corollary~\ref{cor:datadriven}).
\item \emph{Rates and functional central limit theorem under four moments.} We give almost sure rates for the iterates (Theorem~\ref{thm:rate}) and a functional central limit theorem for the partial sums of the iterates, with the efficient GMM variance $(G\tr\Omega^{-1}G)^{-1}$ (Theorem~\ref{thm:fclt}). A blocking argument keeps the moment condition at four moments of the moment function and its Jacobian for step exponents $a>2/3$.
\item \emph{Online inference.} The recursively updated matrix $H_n$ converges almost surely to the asymptotic variance, so plug-in confidence intervals cost nothing extra; random-scaling intervals \cite{LeeLiaoSeoShin2022} are also justified (Corollary~\ref{cor:inference}). A projected, linearised and recursively computable version of the Hansen--Sargan statistic has the usual $\chi^2_{\dg-\db}$ limit (Theorem~\ref{thm:J}); the projection makes it insensitive to the distance between the averaged iterate and the minimiser of the GMM criterion.
\item \emph{Dependent data.} For $L$-dependent data, the same parameter step with long-run variance weights retains the convergence, rate, central limit and inference results with the efficient GMM variance for dependent data, under a stronger moment condition (more than eight moments for $a=3/4$; Theorem~\ref{thm:dependent}).
\item \emph{Evidence.} Simulations for linear and exponential IV models with $1000$ replications and up to $10^6$ observations; size and power of the online test; invariance to the units of measurement; timings of a compiled implementation up to $10^7$ observations; a stochastic volatility model with dependent data (Section~\ref{sec:numerics}); and an empirical application with $254{,}654$ observations (Section~\ref{sec:application}).
\end{enumerate}

\paragraph{Relation to \citet{ChenEtAl2025} and \citet{LeungChanShao2025}.} The two papers closest to ours are these. \citet{ChenEtAl2025} introduced SGMM and developed its theory for \emph{linear} IV models with \emph{i.i.d.\ data}: almost sure convergence of online 2SLS, functional central limit theorems for online 2SLS and for efficient SGMM, plug-in and random-scaling inference, and online Durbin--Wu--Hausman and Hansen--Sargan tests. Their Section~7 proposes the nonlinear algorithm studied here and leaves its asymptotic properties open. Relative to that paper, our contributions are as follows.
\begin{itemize}[leftmargin=*,itemsep=1pt]
\item \emph{Asymptotic theory for nonlinear SGMM.} We prove convergence, rates, a functional central limit theorem with the efficient variance, and inference for nonlinear SGMM. The Jacobian and the weight matrix are re-estimated at the averaged iterate, so the preconditioner depends on the whole path. In the linear SGMM of \citet{ChenEtAl2025}, the Jacobian does not depend on the parameter and the weight is computed at a fixed preliminary estimate. Handling this dependence requires the preconditioned monotonicity condition and the local analysis of Proposition~\ref{prop:local} and Corollary~\ref{cor:datadriven}.
\item \emph{Rank-two update.} We give a rank-two update of the Gauss--Newton matrix that covers nonlinear IV-type moments. For linear efficient SGMM, \citet{ChenEtAl2025} use a rank-three update.
\item \emph{Variance estimate.} We show that the inverse Gauss--Newton matrix $H_n$, which enters the preconditioner, is a strongly consistent variance estimator.
\item \emph{Projected overidentification statistic.} We propose a projected overidentification statistic. It removes the contribution of the distance between the averaged iterate and the minimiser of the GMM criterion, which can dominate the unprojected statistic in finite samples (Section~\ref{sec:application}).
\end{itemize}
\citet{LeungChanShao2025} study online GMM for stationary time series. Their estimator is updated explicitly batch by batch: nonlinear moments are handled by expanding the sample moments around estimates from earlier batches, and the weight matrix is an online long-run variance estimator. They establish efficiency under serial dependence and propose online Hansen--Sargan and structural stability tests. They also point out two weaknesses of SGMM in its original form: sensitivity to the initial learning rate, and size distortion of its online Hansen--Sargan test in finite samples. Our paper differs from theirs in three ways.
\begin{itemize}[leftmargin=*,itemsep=1pt]
\item \emph{Type of algorithm.} Our algorithm is a stochastic approximation that updates the parameter with every observation, and our \emph{dependent-data} result (Theorem~\ref{thm:dependent}) shows that this recursion keeps its efficiency with long-run variance weights. Our dependence assumption, $L$-dependence, is much more restrictive than theirs.
\item \emph{Learning rate.} The Gauss--Newton preconditioner makes the parameter step invariant to invertible linear reparametrizations of the parameter and of the moments, apart from the effect of the ridge and the projection (Remark~\ref{rem:design}(v)). Our default step rule then works without tuning across our examples, whereas a first-order version fails after a change of units (Section~\ref{sec:sens}).
\item \emph{Overidentification test.} Our projected online statistic has close to nominal size in our simulations (Section~\ref{sec:numJ}).
\end{itemize}
A direct numerical comparison with their estimator is left for future work.\\

\paragraph{\textit{Related literature.}} Stochastic approximation goes back to \citet{RobbinsMonro1951}; averaging of the iterates was introduced by \citet{Ruppert1988} and \citet{PolyakJuditsky1992}, and almost sure analyses are given by \citet{Duflo1997}, \citet{Pelletier1998} and \citet{KushnerYin2003}. Stochastic Newton-type algorithms with recursively inverted Hessian estimates were developed for logistic regression \cite{BercuGodichonPortier2020} and nonlinear regression \cite{CenacGodichonPortier2025}, and the use of Sherman--Morrison updates to avoid matrix inversions is also central to the natural-gradient variational Bayes method of \citet{GodichonNguyenTran2025}. Online inference based on the functional central limit theorem was proposed by \citet{LeeLiaoSeoShin2022}. All of these concern the minimization of an expected loss; GMM with overidentification is different, because the estimating function is a fixed linear combination of the moments with an unknown optimal weight, and the weight matrix and the Jacobian must be learned together with the parameter. \\

\paragraph{The rest of this paper is organized as follows:} Section~\ref{sec:model} introduces the model and examples. Section~\ref{sec:algo} presents the algorithm and the recursive inversion formulas. Section~\ref{sec:theory} states the asymptotic results for independent data and Section~\ref{sec:dependent} the extension to dependent data; the proofs are in Section~\ref{sec:proofs}. Section~\ref{sec:numerics} reports simulations, Section~\ref{sec:application} the empirical application, and Section~\ref{sec:conclusion} concludes.\\

\paragraph{\textit{Notation}.} $\|x\|$ is the Euclidean norm of a vector and $\|A\|$ the operator norm of a matrix; $\lambda_{\min}(A)$, $\lambda_{\max}(A)$ are the extreme eigenvalues of a symmetric matrix and $\sigma_{\min}(A)$ the smallest singular value; $A\preceq B$ means that $B-A$ is positive semi-definite; $I_d$ is the $d\times d$ identity. For a closed set $\mathcal S$ of matrices, $\dist(A,\mathcal S)=\min_{B\in\mathcal S}\|A-B\|$. Convergence in distribution is denoted by $\Rightarrow$.

%=====================================================================
\section{Model and examples}\label{sec:model}
%=====================================================================

Let $X,X_1,X_2,\dots$ be independent and identically distributed random elements of a measurable space $\mathcal X$ and $\Theta\subset\R^{\db}$. For $x\in\mathcal X$ and $\beta\in\Theta$ let $g(x,\beta)\in\R^{\dg}$ be the moment function and $G(x,\beta)=\partial g(x,\beta)/\partial\beta\tr\in\R^{\dg\times\db}$ its Jacobian. Write
\[
m(\beta)=\E[g(X,\beta)],\qquad \bar G(\beta)=\E[G(X,\beta)],\qquad \Omega(\beta)=\E\big[g(X,\beta)g(X,\beta)\tr\big],
\]
and $G^*=\bar G(\beta^*)$, $\Omega^*=\Omega(\beta^*)$. The parameter $\beta^*$ is the unique zero of $m$ in $\Theta$. For a positive definite weight matrix $W$, the GMM estimator minimises $\bar g_n(\beta)\tr W\bar g_n(\beta)$ with $\bar g_n(\beta)=n^{-1}\sum_{i=1}^ng(X_i,\beta)$, and under standard conditions \cite{Hansen1982,NeweyMcFadden1994}
\begin{equation}\label{eq:sandwich}
\sqrt n(\hat\beta_n-\beta^*)\Rightarrow N\big(0,\Sigma_W\big),\qquad \Sigma_W=(G^{*\top}WG^*)^{-1}G^{*\top}W\Omega^*WG^*(G^{*\top}WG^*)^{-1}.
\end{equation}
The variance is minimal for $W\propto\Omega^{*-1}$, where it equals $\Sigma^*=(G^{*\top}\Omega^{*-1}G^*)^{-1}$; the two-step estimator achieves it with a preliminary consistent estimate of $\Omega^*$.

\begin{example}[Linear IV regression]\label{ex:linear}
Let $X=(y,x,z)$ with $y=x\tr\beta^*+u$, $\E[zu]=0$, where some components of $x\in\R^{\db}$ are correlated with $u$ and $z\in\R^{\dg}$ is a vector of instruments. Then $g(X,\beta)=z(x\tr\beta-y)$, $G(X,\beta)=zx\tr$ does not depend on $\beta$, and $\bar G=\E[zx\tr]$ has full column rank under the usual relevance condition. With $W=\E[zz\tr]^{-1}$, \eqref{eq:sandwich} is the asymptotic variance of 2SLS; under heteroskedasticity the efficient GMM estimator is in general more efficient.
\end{example}

\begin{example}[Nonlinear IV regression]\label{ex:nonlinear}
More generally let $g(X,\beta)=z\,r(X,\beta)$ for a scalar residual $r$. Examples are the additive model $y=f(x,\beta^*)+u$ with $\E[zu]=0$, where $r=f(x,\beta)-y$, and the multiplicative exponential model for non-negative or count outcomes, $y=\exp(x\tr\beta^*)\eta$ with $\E[\eta\mid z]=1$ \cite{Mullahy1997,WindmeijerSantosSilva1997}, where
\[
r(X,\beta)=y\,e^{-x\tr\beta}-1,\qquad \partial_\beta r(X,\beta)=-y\,e^{-x\tr\beta}x .
\]
In both cases the Jacobian $G(X,\beta)=z\,\partial_\beta r(X,\beta)\tr$ has rank one, which is exploited in Lemma~\ref{lem:woodbury}. The linear model of Example~\ref{ex:linear} is the special case $r=x\tr\beta-y$.
\end{example}

\begin{example}[Stochastic volatility]\label{ex:sv}
Consider the log-normal stochastic volatility model \cite{Taylor2008}
\[
y_i=\sigma_iZ_i,\qquad \ln\sigma_i^2=\omega+b\ln\sigma_{i-1}^2+\sigma_u u_i,\qquad (Z_i,u_i)\tr\sim N(0,I_2)\ \text{i.i.d.},
\]
with $|b|<1$ and parameter $\beta^*=(\omega,b,\sigma_u)\tr$. The log-volatility $h_i=\ln\sigma_i^2$ is a stationary Gaussian AR(1) process with mean $\mu=\omega/(1-b)$ and variance $\tau^2=\sigma_u^2/(1-b^2)$, so that $\E[\sigma_i^r]=\exp(r\mu/2+r^2\tau^2/8)$ and, since $\operatorname{Cov}(h_i,h_{i-j})=b^j\tau^2$,
\begin{equation}\label{eq:svcross}
\E\big[\sigma_i^r\sigma_{i-j}^s\big]=\E[\sigma_i^r]\,\E[\sigma_i^s]\exp\Big(\frac{rs\,b^j\tau^2}{4}\Big).
\end{equation}
Following \citet{JacquierPolsonRossi1994} and \citet{AndersenSorensen1996} we use the $\dg=24$ moments
\[
\begin{gathered}
\E|y_i|=\sqrt{2/\pi}\,\E[\sigma_i],\quad \E[y_i^2]=\E[\sigma_i^2],\quad \E|y_i|^3=2\sqrt{2/\pi}\,\E[\sigma_i^3],\quad \E[y_i^4]=3\E[\sigma_i^4],\\
\E|y_iy_{i-j}|=\tfrac2\pi\E[\sigma_i\sigma_{i-j}],\qquad \E[y_i^2y_{i-j}^2]=\E[\sigma_i^2\sigma_{i-j}^2],\qquad j=1,\dots,10,
\end{gathered}
\]
that is, $g(X_i,\beta)=x_i-\mu_{\rm SV}(\beta)$ with $x_i=(|y_i|,y_i^2,|y_i|^3,y_i^4,|y_iy_{i-1}|,\dots,|y_iy_{i-10}|,y_i^2y_{i-1}^2,\dots,y_i^2y_{i-10}^2)\tr$ and $\mu_{\rm SV}(\beta)$ the vector of the right-hand sides computed with \eqref{eq:svcross}. Here the Jacobian $G(X_i,\beta)=-\partial\mu_{\rm SV}(\beta)/\partial\beta\tr$ does not depend on the data. The data $X_i=(y_i,\dots,y_{i-10})$ are stationary and geometrically mixing but not independent, and the efficient weight is the inverse of the long-run variance of the moments (Section~\ref{sec:dependent}); we use this example in Section~\ref{sec:sv}.
\end{example}

\begin{remark}
The exponent in \eqref{eq:svcross} follows from $\operatorname{Var}(rh_i/2+sh_{i-j}/2)=(r^2+s^2)\tau^2/4+rs\,b^j\tau^2/2$; for $r=s=2$ it gives $\E[\sigma_i^2\sigma_{i-j}^2]=\E[\sigma_i^2]^2\exp(b^j\tau^2)$. We verified all $24$ moment formulas against a simulated path of length $2\times10^6$.
\end{remark}

%=====================================================================
\section{The algorithm}\label{sec:algo}
%=====================================================================

\subsection{Stochastic GMM with Gauss--Newton preconditioning}
The first-order condition of the GMM problem with weight $W$ is $\bar G_n(\beta)\tr W\bar g_n(\beta)=0$. Its population version $\bar G(\beta)\tr Wm(\beta)=0$ suggests the stochastic approximation recursion $\beta_i=\beta_{i-1}-\gamma_iP\,g(X_i,\beta_{i-1})$ with a preconditioning matrix $P\in\R^{\db\times\dg}$. The choice
\begin{equation}\label{eq:P}
P(\Phi,W)=(\Phi\tr W\Phi)^{-1}\Phi\tr W,\qquad\text{so that}\quad P(\Phi,W)\Phi=I_{\db},
\end{equation}
with $\Phi\approx G^*$ and $W\approx\Omega^{*-1}$, makes the linearised mean field $\beta\mapsto P(G^*,W)m(\beta)\approx\beta-\beta^*$ near $\beta^*$: it is a stochastic Gauss--Newton step. Since $\Phi$ and $W$ are unknown, they are estimated recursively along the way.

Let $n_0$ pilot observations be available (in a streaming context, the first $n_0$ observations). Let $\Theta\subset\R^{\db}$ be a compact convex set with $\beta^*$ in its interior, $\Pi_\Theta$ the Euclidean projection onto $\Theta$, and $\eta_0>0$ a small ridge constant.

\begin{algorithm}[t]
\caption{Stochastic GMM with Gauss--Newton preconditioning (SGMM)}\label{alg:sgmm}
\begin{algorithmic}[1]
\Require pilot sample $X_{-n_0+1},\dots,X_0$; steps $\gamma_i=\gamma_0(i+i_0)^{-a}$; truncation levels $\bar p_i=c_Pi^{\nu}$
\State \textbf{Initialization} ($n_0\ge\db$): $\beta_0=\Pi_\Theta(\hat\beta_{\rm pilot})$, e.g.\ one-step GMM on the pilot sample; $\bar\beta_0=\beta_0$;
\Statex\hspace{\algorithmicindent} $\Phi_0=\frac1{n_0}\sum_{j\le0}G(X_j,\beta_0)$, \ $S_0=\frac1{n_0}\sum_{j\le0}g(X_j,\beta_0)g(X_j,\beta_0)\tr+\eta_0I_{\dg}$, \ $W_0=S_0^{-1}$,
\Statex\hspace{\algorithmicindent} $H_0=(\Phi_0\tr W_0\Phi_0)^{-1}$
\For{$i=1,2,\dots$}
\State $\hat P=H_{i-1}\Phi_{i-1}\tr W_{i-1}$; \ $P_{i-1}=\hat P\,\mathbf 1\{\Phi_{i-1}\tr W_{i-1}\Phi_{i-1}\text{ invertible},\ \|\hat P\|\le\bar p_i\}$ \Comment{truncation}
\State $\beta_i=\Pi_\Theta\big(\beta_{i-1}-\gamma_iP_{i-1}\,g(X_i,\beta_{i-1})\big)$ \Comment{preconditioned step}
\State $s_i=g(X_i,\bar\beta_{i-1})$, \ $\Phi_i=\Phi_{i-1}+\frac{1}{n_0+i}\big(G(X_i,\bar\beta_{i-1})-\Phi_{i-1}\big)$ \Comment{Jacobian}
\State $\ell_i=n_0+i-1+s_i\tr W_{i-1}s_i$, \ $W_i=\frac{n_0+i}{n_0+i-1}\Big(W_{i-1}-\frac1{\ell_i}W_{i-1}s_is_i\tr W_{i-1}\Big)$ \Comment{inverse covariance}
\State $H_i=(\Phi_i\tr W_i\Phi_i)^{-1}$ if invertible \Comment{by Lemma~\ref{lem:woodbury} for IV-type moments}
\State $\bar\beta_i=\bar\beta_{i-1}+\frac1i(\beta_i-\bar\beta_{i-1})$ \Comment{averaging}
\EndFor
\State \textbf{Output:} $\bar\beta_n$, variance estimate $H_n/n$, statistic $\mathcal J_n$ of \eqref{eq:Jn}
\end{algorithmic}
\end{algorithm}

\begin{remark}[Design choices]\label{rem:design}
(i) The parameter step uses the current iterate $\beta_{i-1}$, while the Jacobian and the weight matrix are evaluated at the averaged iterate $\bar\beta_{i-1}$, which is less noisy; both are $\F_{i-1}$-measurable, where $\F_i$ is generated by the pilot sample and $X_1,\dots,X_i$, so that every update is driven by a martingale difference. This is the convention of the nonlinear SGMM of \citet[Section~7]{ChenEtAl2025}; for their linear efficient SGMM they compute the weight at a fixed preliminary estimate instead.
(ii) By construction $\Phi_i=(n_0+i)^{-1}\big(n_0\Phi_0+\sum_{j=1}^iG(X_j,\bar\beta_{j-1})\big)$ and $W_i=S_i^{-1}$ with $S_i=(n_0+i)^{-1}\big(n_0S_0+\sum_{j=1}^is_js_j\tr\big)$; line~6 is the Sherman--Morrison formula for this inverse (Lemma~\ref{lem:woodbury}(a)).
(iii) The projection and the truncation are needed for the proofs; they are inactive after a finite random time (Theorem~\ref{thm:as}). If $\Phi_i\tr W_i\Phi_i$ is singular, which almost surely happens only finitely often (Lemma~\ref{lem:precond}), the step is skipped and $H_i$ is computed directly once the matrix becomes invertible. In our simulations of independent data the truncation and the projection onto a large box were never active for SGMM with the default steps of (v) below (they were for the first-order version in rescaled units; Section~\ref{sec:sens}).
(iv) Replacing $s_i$ by the instrument vector $z_i$ in line~6 yields $W_i\to\E[zz\tr]^{-1}$ and the online 2SLS estimator of \citet{ChenEtAl2025} in the linear case (Remark~\ref{rem:2sls}).
(v) \emph{Default steps.} We use $a=3/4$, $i_0=n_0$ and $\gamma_0=c_\gamma n_0^{a-1}$, i.e.\ $\gamma_i=c_\gamma n_0^{a-1}(n_0+i)^{-a}$ with $c_\gamma=2$. The first step is then $\gamma_1\approx c_\gamma/n_0$, the weight of one additional observation in a recursive estimate based on the $n_0$ pilot observations, and the steps continue to decrease as $(n_0+i)^{-a}$. This rule needs no tuning to the scale of the problem, because the Gauss--Newton step is invariant to invertible linear reparametrizations of $\beta$ and of the moments, apart from the effect of the ridge $\eta_0$ and of the projection (Section~\ref{sec:sens}). It matters most when some parameters are weakly identified, so that single observations produce large Gauss--Newton steps: in the application of Section~\ref{sec:application}, larger first steps ($i_0=100$, $\gamma_0=1$) left the final estimates more than three standard errors away from the offline estimates in all $20$ orderings of the data. The offset $i_0=n_0$ also damps the first steps, in which a single observation would otherwise move the iterate by almost a full Gauss--Newton step.
\end{remark}

\subsection{Recursive inversion}
Each iteration of Algorithm~\ref{alg:sgmm} requires $H_i=(\Phi_i\tr W_i\Phi_i)^{-1}$. Computing $\Phi_i\tr W_i\Phi_i$ directly costs $O(\dg^2\db+\dg\db^2)$ operations, and inverting it $O(\db^3)$. For IV-type moments the following lemma updates $H_i$ by inverting a $2\times2$ matrix only.

\begin{lemma}[Recursive inversion]\label{lem:woodbury}
Fix $i\ge1$, write $N=n_0+i-1$, $\Phi=\Phi_{i-1}$, $W=W_{i-1}$, $H=H_{i-1}$, $s=s_i$, and suppose that $\Phi\tr W\Phi$ and $\Phi_i\tr W_i\Phi_i$ are invertible.
\begin{enumerate}[label=(\alph*),itemsep=2pt]
\item With $\ell=N+s\tr Ws$, the inverse of $S_i=(NS_{i-1}+ss\tr)/(N+1)$ is $W_i=\frac{N+1}{N}\big(W-\frac1\ell Wss\tr W\big)$.
\item Suppose that $G(X_i,\bar\beta_{i-1})=z\,\psi\tr$ and $s=\varrho z$ with $z\in\R^{\dg}$, $\psi\in\R^{\db}$ and $\varrho\in\R$. For IV-type moments $g=z\,r$ this holds with $z=z_i$, $\psi=\partial_\beta r(X_i,\bar\beta_{i-1})$ and $\varrho=r(X_i,\bar\beta_{i-1})$; for 2SLS weights (Remark~\ref{rem:design}(iv)) it holds with $\varrho=1$. Let
\[
u=\Phi\tr Wz,\qquad c_z=z\tr Wz,\qquad \ell=N+\varrho^2c_z,\qquad U=[\,u,\ \psi\,]\in\R^{\db\times2},\qquad \mathcal C=\frac1\ell\begin{pmatrix}-N\varrho^2&N\\N&c_z\end{pmatrix}.
\]
Then $I_2+N^{-1}\mathcal CU\tr HU$ is invertible and
\begin{equation}\label{eq:woodbury3}
H_i=\frac{N+1}{N}\Big(H-\frac1N\,HU\big(I_2+N^{-1}\mathcal CU\tr HU\big)^{-1}\mathcal CU\tr H\Big).
\end{equation}
\item (Linear 2SLS.) If $\varrho=1$, then, with $U_2=[\,\psi-u,\ \psi\,]$ and $D=\diag(-\ell,N)$,
\begin{equation}\label{eq:woodbury2}
H_i=\frac{N+1}{N}\Big(H-HU_2\big(D+U_2\tr HU_2\big)^{-1}U_2\tr H\Big).
\end{equation}
\end{enumerate}
\end{lemma}

The cost of \eqref{eq:woodbury3} is $O(\dg^2+\dg\db+\db^2)$: $Wz$ costs $O(\dg^2)$, $u$ costs $O(\dg\db)$, and the rest $O(\db^2)$. Together with the $O(\dg^2)$ update of $W_i$ and the $O(\dg\db)$ update of $\Phi_i$, one iteration of Algorithm~\ref{alg:sgmm} costs $O(\dg^2+\dg\db+\db^2)$ operations and the memory is $O(\dg^2+\dg\db)$, independently of $n$. Formula \eqref{eq:woodbury2} is the update used by \citet[eq.~(3)]{ChenEtAl2025} for online 2SLS, and they use a rank-three Woodbury update for their linear efficient SGMM; part (b) shows that a rank-two update suffices for all IV-type moments, linear or not, with continuously updated efficient weights. We checked \eqref{eq:woodbury3} and \eqref{eq:woodbury2} numerically against direct inversion (relative errors below $10^{-15}$).

\begin{proof}
(a) is the Sherman--Morrison formula applied to $NS_{i-1}+ss\tr$, whose inverse is $N^{-1}W-N^{-2}Wss\tr W/(1+N^{-1}s\tr Ws)=N^{-1}(W-Wss\tr W/\ell)$, multiplied by $N+1$.

(b) Put $A=N\Phi+z\psi\tr$, so that $\Phi_i=A/(N+1)$ and, by (a), $W_i=(N+1)(NS_{i-1}+ss\tr)^{-1}=\frac{N+1}{N}(W-Wss\tr W/\ell)$, where $s\tr Ws=\varrho^2c_z$. Hence
\[
\Phi_i\tr W_i\Phi_i=\frac{1}{N(N+1)}\Big(A\tr WA-\frac1\ell(A\tr Ws)(A\tr Ws)\tr\Big).
\]
Now $A\tr WA=N^2\Phi\tr W\Phi+N(u\psi\tr+\psi u\tr)+c_z\psi\psi\tr$ and $A\tr Ws=\varrho(Nu+c_z\psi)$, so
\[
\begin{aligned}
A\tr WA-\frac{(A\tr Ws)(A\tr Ws)\tr}{\ell}&=N^2\Phi\tr W\Phi+N\Big(1-\frac{\varrho^2c_z}{\ell}\Big)(u\psi\tr+\psi u\tr)\\&\quad-\frac{N^2\varrho^2}{\ell}uu\tr+c_z\Big(1-\frac{\varrho^2c_z}{\ell}\Big)\psi\psi\tr\\
&=N^2\Phi\tr W\Phi+N\,U\mathcal CU\tr ,
\end{aligned}
\]
because $1-\varrho^2c_z/\ell=N/\ell$ and $U\mathcal CU\tr=\ell^{-1}\big(-N\varrho^2uu\tr+N(u\psi\tr+\psi u\tr)+c_z\psi\psi\tr\big)$. Therefore $(N+1)\Phi_i\tr W_i\Phi_i=N\,\Phi\tr W\Phi+U\mathcal CU\tr$. For an invertible matrix $B$ and any $U$, $\mathcal C$ such that $B+U\mathcal CU\tr$ is invertible, the matrix $I+\mathcal CU\tr B^{-1}U$ is invertible and
\[
(B+U\mathcal CU\tr)^{-1}=B^{-1}-B^{-1}U(I+\mathcal CU\tr B^{-1}U)^{-1}\mathcal CU\tr B^{-1}.
\]
Indeed, if $(I+\mathcal CU\tr B^{-1}U)v=0$, then $w=B^{-1}Uv$ satisfies $(B+U\mathcal CU\tr)w=U(I+\mathcal CU\tr B^{-1}U)v=0$, so $w=0$, hence $Uv=Bw=0$ and $v=-\mathcal CU\tr B^{-1}Uv=0$; the identity is then verified by multiplying the right-hand side by $B+U\mathcal CU\tr$. With $B=N\Phi\tr W\Phi$, $B^{-1}=H/N$, this gives \eqref{eq:woodbury3}.

(c) For $\varrho=1$, $\ell=N+c_z$ and, expanding the right-hand side and using $1-N/\ell=c_z/\ell$,
\[
\begin{aligned}U\mathcal CU\tr&=\frac1\ell\big(-Nuu\tr+N(u\psi\tr+\psi u\tr)+c_z\psi\psi\tr\big)\\&=-\frac{N}{\ell}(\psi-u)(\psi-u)\tr+\psi\psi\tr=U_2\tilde{\mathcal C}U_2\tr,\qquad\tilde{\mathcal C}=\diag(-N/\ell,1).\end{aligned}
\]
Hence $(N+1)\Phi_i\tr W_i\Phi_i=NH^{-1}+U_2\tilde{\mathcal C}U_2\tr$, and the classical Woodbury identity
\[
\big(NH^{-1}+U_2\tilde{\mathcal C}U_2\tr\big)^{-1}=\frac1N\Big(H-HU_2\big(N\tilde{\mathcal C}^{-1}+U_2\tr HU_2\big)^{-1}U_2\tr H\Big),\qquad N\tilde{\mathcal C}^{-1}=\diag(-\ell,N)=D,
\]
gives \eqref{eq:woodbury2}; the $2\times2$ matrix $D+U_2\tr HU_2$ is invertible by the argument of part (b).
\end{proof}

\subsection{Online overidentification statistic}
The Hansen--Sargan statistic $n\,\bar g_n(\hat\beta)\tr W\bar g_n(\hat\beta)$ requires the moments of all observations at the final estimate. A first-order expansion around the averaged iterates gives a recursively computable substitute: let
\begin{equation}\label{eq:Jn}
\hat g_n=\frac1n\sum_{i=1}^n\Big(g(X_i,\bar\beta_{i-1})+G(X_i,\bar\beta_{i-1})(\bar\beta_n-\bar\beta_{i-1})\Big),\qquad \mathcal J_n=n\,\hat g_n\tr\big(W_n-W_n\Phi_nH_n\Phi_n\tr W_n\big)\hat g_n .
\end{equation}
Since $\hat g_n=n^{-1}\big(\sum_ig(X_i,\bar\beta_{i-1})-\sum_iG(X_i,\bar\beta_{i-1})\bar\beta_{i-1}\big)+\big(n^{-1}\sum_iG(X_i,\bar\beta_{i-1})\big)\bar\beta_n$, three running sums of total size $O(\dg\db)$ suffice.

The matrix $W_n-W_n\Phi_nH_n\Phi_n\tr W_n$ projects out the directions $\Phi_nd$, $d\in\R^{\db}$. Because $\hat g_n+\Phi_nd$ approximates the sample moments at $\bar\beta_n+d$, $\mathcal J_n$ approximates the minimum over $\beta$ of the GMM criterion near $\bar\beta_n$, which is what the Hansen--Sargan statistic measures. It does not require $\bar\beta_n$ to be the exact minimiser. The unprojected statistic $n\hat g_n\tr W_n\hat g_n$ has the same limit (see the proof of Theorem~\ref{thm:J}), but in finite samples it also measures the distance between $\bar\beta_n$ and the minimiser; this can be large in well-identified linear combinations of the parameters, and Section~\ref{sec:application} gives an example.

%=====================================================================
\section{Asymptotic theory}\label{sec:theory}
%=====================================================================

\subsection{Assumptions}
Let $\F_0$ be the $\sigma$-field generated by the pilot sample and $\F_i=\sigma(\F_0,X_1,\dots,X_i)$.

\begin{assumption}[Data and parameter space]\label{ass:data}
$X_1,X_2,\dots$ are i.i.d.\ copies of $X$ and independent of $\F_0$. The set $\Theta\subset\R^{\db}$ is compact and convex, $\beta^*$ lies in its interior, $n_0\ge\db$, and $\beta_0\in\Theta$, $\Phi_0$ and the positive definite $S_0$ are $\F_0$-measurable.
\end{assumption}

\begin{assumption}[Smoothness and moments]\label{ass:smooth}
For every $\beta$, $x\mapsto g(x,\beta)$ is measurable, and for every $x$, $\beta\mapsto g(x,\beta)$ is continuously differentiable on an open set $\mathcal O\supset\Theta$. Let $B(x)=\sup_{\beta\in\Theta}\|g(x,\beta)\|+\sup_{\beta\in\Theta}\|G(x,\beta)\|$. There is a measurable $L$ with $\|G(x,\beta)-G(x,\beta')\|\le L(x)\|\beta-\beta'\|$ for $\beta,\beta'\in\Theta$, and $\E[B(X)^\kappa]<\infty$, $\E[L(X)]<\infty$ for some $\kappa\ge4$. Moreover $\lambda_0=\inf_{\beta\in\Theta}\lambda_{\min}(\Omega(\beta))>0$ and $G^*$ has full column rank $\db$.
\end{assumption}

Under Assumption~\ref{ass:smooth}, $m$ and $\bar G$ are continuous, $\bar G$ is Lipschitz on $\Theta$ with constant $\bar L=\E[L(X)]$, $m$ is continuously differentiable on $\operatorname{int}\Theta$ with derivative $\bar G$ (by dominated convergence), and $\Omega$ is Lipschitz on $\Theta$ with constant $2\E[B(X)^2]$, because $\|g(x,\beta)-g(x,\beta')\|\le B(x)\|\beta-\beta'\|$ by the mean value theorem on the convex set $\Theta$. Let
\[
\Lambda=\sup_{\beta\in\Theta}\lambda_{\max}(\Omega(\beta)),\qquad \mathcal G=\conv\{\bar G(\beta):\beta\in\Theta\},\qquad \mathcal W=\big\{W=W\tr:\ (2\Lambda)^{-1}I\preceq W\preceq 2\lambda_0^{-1}I\big\},
\]
where $\mathcal G$ is compact as the convex hull of a compact set in a finite-dimensional space, and $\mathcal W$ is closed and bounded.

\begin{assumption}[Preconditioned monotonicity]\label{ass:mono}
There are $\delta_\Phi,\sigma_0,c>0$ such that every $\Phi\in\R^{\dg\times\db}$ with $\dist(\Phi,\mathcal G)\le\delta_\Phi$ satisfies $\sigma_{\min}(\Phi)\ge\sigma_0$, and
\[
\big\langle\beta-\beta^*,\,P(\Phi,W)\,m(\beta)\big\rangle\ge c\,\|\beta-\beta^*\|^2\qquad\text{for all }\beta\in\Theta,\ W\in\mathcal W,\ \dist(\Phi,\mathcal G)\le\delta_\Phi,
\]
with $P(\Phi,W)$ as in \eqref{eq:P}.
\end{assumption}

\begin{assumption}[Steps and truncation]\label{ass:steps}
$\gamma_i=\gamma_0(i+i_0)^{-a}$ with $\gamma_0>0$, $i_0\ge0$ and $a\in(1/2,1)$; $\bar p_i=c_Pi^\nu$ with $c_P>0$ and $0<\nu<a-1/2$.
\end{assumption}

Assumption~\ref{ass:mono} implies that $\beta^*$ is the only zero of $m$ in $\Theta$. It is the condition under which the preconditioned mean field points towards $\beta^*$ for all preconditioners that can occur along the algorithm. In the linear model the Jacobian does not depend on $\beta$, $\mathcal G=\{G^*\}$, and the condition holds on every $\Theta$; for smooth nonlinear models it holds on every sufficiently small ball around $\beta^*$:

\begin{proposition}[Validity of Assumption~\ref{ass:mono}]\label{prop:local}
Suppose that Assumption~\ref{ass:smooth} holds on a closed ball $\Theta_0=\bar B(\beta^*,r_0)$.
\begin{enumerate}[label=(\alph*),itemsep=2pt]
\item There are $r_1\in(0,r_0]$ and $\delta_\Phi>0$ such that Assumption~\ref{ass:mono} holds with $c=1/2$ on every compact convex $\Theta\subseteq\bar B(\beta^*,r_1)$ with $\beta^*\in\operatorname{int}\Theta$.
\item In the linear model of Example~\ref{ex:linear}, if Assumption~\ref{ass:smooth} holds on a compact convex $\Theta$, then Assumption~\ref{ass:mono} holds on $\Theta$ with $c=1/2$, for some $\delta_\Phi>0$.
\end{enumerate}
\end{proposition}

In practice, $\beta^*$ is unknown and $\Theta$ would be chosen as a ball around the pilot estimate, which is close to $\beta^*$ when $n_0$ is moderately large. Such an $\F_0$-measurable $\Theta$ is covered by our results: since $X_1,X_2,\dots$ are independent of $\F_0$, all theorems below hold conditionally on $\F_0$, on the event that $\Theta$ satisfies Assumptions~\ref{ass:data}--\ref{ass:mono}. Combined with Proposition~\ref{prop:local}(a) this gives a result that requires no knowledge of $\beta^*$:

\begin{corollary}[Data-driven parameter set]\label{cor:datadriven}
Suppose Assumption~\ref{ass:smooth} holds on $\Theta_0=\bar B(\beta^*,r_0)$, let $r_1$ be as in Proposition~\ref{prop:local}(a), $0<\rho\le r_1/2$, and let $\Theta=\bar B(\hat\beta_{\rm pilot},\rho)$ for an $\F_0$-measurable pilot estimator. On the event $\mathcal E_\rho=\{\|\hat\beta_{\rm pilot}-\beta^*\|<\rho\}$, $\beta^*\in\operatorname{int}\Theta$ and $\Theta\subseteq\bar B(\beta^*,2\rho)\subseteq\bar B(\beta^*,r_1)$, so that the conclusions of Theorems~\ref{thm:as}--\ref{thm:J} below hold, under their respective moment conditions, conditionally on $\F_0$ on $\mathcal E_\rho$. In particular $\Prob(\bar\beta_n\to\beta^*)\ge\Prob(\mathcal E_\rho)$, which tends to one as $n_0\to\infty$ for a consistent pilot estimator.
\end{corollary}

A global version of Assumption~\ref{ass:mono} for a given nonlinear model is a condition on the variation of the Jacobian over $\Theta$, of the same nature as the conditions used for stochastic Gauss--Newton algorithms in nonlinear regression \cite{CenacGodichonPortier2025}.

\subsection{Main results}

\begin{theorem}[Almost sure convergence]\label{thm:as}
Under Assumptions~\ref{ass:data}--\ref{ass:steps}, $\beta_n\to\beta^*$ and $\bar\beta_n\to\beta^*$ almost surely. Moreover, almost surely, $\Phi_n\to G^*$, $W_n\to\Omega^{*-1}$, $H_n\to\Sigma^*=(G^{*\top}\Omega^{*-1}G^*)^{-1}$ and $P_n\to P^*=\Sigma^*G^{*\top}\Omega^{*-1}$, and there is a finite random time after which neither the truncation nor the projection is active.
\end{theorem}

Let
\[
\kappa(a)=\max\Big\{2,\ \frac{4(1-a)}{2a-1}\Big\}.
\]
Since $\kappa\ge4$ in Assumption~\ref{ass:smooth}, the condition $\kappa>\kappa(a)$ used below holds automatically when $a>2/3$; for the default $a=3/4$, $\kappa(a)=2$.

\begin{theorem}[Almost sure rates]\label{thm:rate}
Under Assumptions~\ref{ass:data}--\ref{ass:steps}, for every $\varepsilon>(1-a)/\kappa$,
\[
\|\beta_n-\beta^*\|=O\big(n^{-a/2+\varepsilon}\big),\qquad\|\bar\beta_n-\beta^*\|=O\big(n^{-a/2+\varepsilon}\big),\qquad \|\Phi_n-G^*\|+\|W_n-\Omega^{*-1}\|=O\big(n^{-a/2+\varepsilon}\big)
\]
almost surely.
\end{theorem}

The rate $n^{-a/2}$ of the non-averaged iterates is the usual one for steps $n^{-a}$; the loss $n^\varepsilon$ can be made arbitrarily small when all moments of $B(X)$ are finite, and could be reduced to a logarithmic factor with the techniques of \citet{Pelletier1998}. The averaged estimator converges faster, at the rate $n^{-1/2}$, as shown next.

\begin{theorem}[Functional central limit theorem]\label{thm:fclt}
Suppose Assumptions~\ref{ass:data}--\ref{ass:steps} hold with $\kappa>\kappa(a)$. Then the partial-sum process
\[
\mathbb B_n(t)=\frac1{\sqrt n}\sum_{i=1}^{\lfloor nt\rfloor}(\beta_i-\beta^*),\qquad t\in[0,1],
\]
converges in distribution in the Skorokhod space $D[0,1]^{\db}$ to $\Sigma^{*1/2}\mathbf B$, where $\mathbf B$ is a standard $\db$-dimensional Brownian motion. More precisely,
\begin{equation}\label{eq:bahadur}
\sup_{t\in[0,1]}\Big\|\mathbb B_n(t)+P^*\frac1{\sqrt n}\sum_{i=1}^{\lfloor nt\rfloor}g(X_i,\beta^*)\Big\|\to0\quad\text{in probability.}
\end{equation}
In particular $\sqrt n(\bar\beta_n-\beta^*)\Rightarrow N(0,\Sigma^*)$, and $\bar\beta_n$ is asymptotically equivalent to the efficient two-step GMM estimator.
\end{theorem}

\begin{remark}[Online 2SLS, other weights and the role of the preconditioner]\label{rem:2sls}
If $s_i$ is replaced by $z_i$ (Remark~\ref{rem:design}(iv)), then $W_n\to\bar W=\E[zz\tr]^{-1}$ and the averaged estimator is asymptotically equivalent to 2SLS-type GMM with weight $\bar W$; in the linear model this is the online 2SLS estimator of \citet{ChenEtAl2025}. Our proofs carry over after replacing $\Omega(\beta)$ by $\E[zz\tr]$ in the definition of $\mathcal W$ and $P^*$ by $P(G^*,\bar W)$, and the limit in Theorem~\ref{thm:fclt} then has the sandwich variance $\Sigma_{\bar W}$ of \eqref{eq:sandwich}. Note that in this case $H_n\to(G^{*\top}\bar WG^*)^{-1}\ne\Sigma_{\bar W}$ under heteroskedasticity, so the plug-in intervals of Corollary~\ref{cor:inference}(a) must use the sandwich $H_n\Phi_n\tr W_n\hat\Omega_nW_n\Phi_nH_n$ with a running estimate $\hat\Omega_n$ of $\Omega^*$ instead; the random-scaling intervals remain valid without change.

The efficiency of $\bar\beta_n$ is due to the limit of the weight matrix: replacing the preconditioner $H_{i-1}\Phi_{i-1}\tr W_{i-1}$ by $K_{i-1}\Phi_{i-1}\tr W_{i-1}$ with any invertible $K_{i-1}$ converging to an invertible limit $\bar K$ (for which the corresponding mean field is stable) leads to the same asymptotic variance after averaging \cite{PolyakJuditsky1992}. In particular the first-order version with $K_{i-1}=I$ is asymptotically equivalent to SGMM. The Gauss--Newton choice $K_{i-1}=H_{i-1}$ makes the linearised recursion $\beta_i-\beta^*\approx(1-\gamma_i)(\beta_{i-1}-\beta^*)$, with a contraction that does not depend on the scale or the conditioning of the problem; this makes the algorithm much less sensitive to $\gamma_0$ (Section~\ref{sec:sens}) and provides the variance estimate $H_n$.
\end{remark}

\begin{corollary}[Online inference]\label{cor:inference}
Under the assumptions of Theorem~\ref{thm:fclt}:
\begin{enumerate}[label=(\alph*),itemsep=2pt]
\item $H_n\to\Sigma^*$ almost surely, and for every $k\le\db$ and $\alpha\in(0,1)$ the interval $\bar\beta_{n,k}\pm z_{1-\alpha/2}\sqrt{H_{n,kk}/n}$ has asymptotic coverage $1-\alpha$.
\item (Random scaling.) Let $V_n=n^{-2}\sum_{s=1}^ns^2(\bar\beta_s-\bar\beta_n)(\bar\beta_s-\bar\beta_n)\tr$. Then
\[
\frac{\sqrt n(\bar\beta_{n,k}-\beta^*_k)}{\sqrt{V_{n,kk}}}\Rightarrow\frac{\mathbf b(1)}{\big(\int_0^1(\mathbf b(r)-r\mathbf b(1))^2\dd r\big)^{1/2}},
\]
with $\mathbf b$ a standard Brownian motion. The limit is pivotal and symmetric, with distribution tabulated by \citet{AbadirParuolo1997} and \citet{KieferVogelsangBunzel2000}; the two-sided $95\%$ critical value is $6.747$ \cite{LeeLiaoSeoShin2022}.
\end{enumerate}
\end{corollary}

\begin{theorem}[Online overidentification test]\label{thm:J}
Suppose the assumptions of Theorem~\ref{thm:fclt} hold and $\dg>\db$. Then $\sqrt n\big(\hat g_n-\bar g_n(\bar\beta_n)\big)\to0$ almost surely, and $\mathcal J_n\Rightarrow\chi^2_{\dg-\db}$.
\end{theorem}

When the moment restrictions \eqref{eq:moment} have no solution, the theorems do not apply. Heuristically, if $\bar\beta_n$ still converges to some pseudo-true value $\beta^\dagger$ with $m(\beta^\dagger)\ne0$ (which requires a separate argument, since the limiting preconditioner then depends on $\beta^\dagger$), then $\hat g_n\to m(\beta^\dagger)\neq0$, so that $\mathcal J_n/n$ has a positive limit and $\mathcal J_n$ grows linearly in $n$; this is what we observe in Section~\ref{sec:numJ}.

%=====================================================================
\section{Dependent data}\label{sec:dependent}
%=====================================================================

Many GMM applications involve time series. Two features of Algorithm~\ref{alg:sgmm} rely on independence: the efficient weight is the inverse of the long-run variance of the moments, not of $\E[gg\tr]$, and the noise $g(X_i,\beta_{i-1})-m(\beta_{i-1})$ is a martingale difference only because $X_i$ is independent of the past iterates. We show that, for $L$-dependent data, only the weight has to be changed.

\begin{assumption}[$L$-dependent data]\label{ass:ldep}
The sequence $(X_i)_{i>-n_0}$ (pilot sample and stream) is strictly stationary and $L$-dependent for some known $L\ge0$: $\sigma(X_j,\,j\le i)$ is independent of $\sigma(X_j,\,j\ge i+L+1)$ for every $i$. The long-run second-moment matrix $\Omega_{\rm LR}(\beta)=\sum_{h=-L}^{L}\E\big[g(X_0,\beta)g(X_h,\beta)\tr\big]$ satisfies $\inf_{\beta\in\Theta}\lambda_{\min}(\Omega_{\rm LR}(\beta))>0$.
\end{assumption}

\paragraph{Long-run weights.} Replace line~6 of Algorithm~\ref{alg:sgmm} by the recursion
\begin{equation}\label{eq:lrv}
S_i=S_{i-1}+\frac{1}{n_0+i}\big(y_i-S_{i-1}\big),\qquad y_i=s_is_i\tr+\sum_{h=1}^{L}k_h\big(s_is_{i-h}\tr+s_{i-h}s_i\tr\big),
\end{equation}
with $k_h=1$, $s_j=g(X_j,\beta_0)$ for the pilot indices $-n_0<j\le0$, and an $\F_0$-measurable positive definite $S_0$ (for instance the pilot version of \eqref{eq:lrv} plus a ridge $\eta_0I$). Set $W_i=S_i^{-1}$ if $S_i$ is positive definite and $W_i=W_{i-1}$ otherwise. With $\check s_i=\sum_{h=1}^Lk_hs_{i-h}$ we have $y_i=(s_i+\check s_i)(s_i+\check s_i)\tr-\check s_i\check s_i\tr$. Hence, when $S_{i-1}$ and $S_i$ are both positive definite, $W_i$ is obtained from $W_{i-1}$ by two Sherman--Morrison steps at cost $O(\dg^2+L\dg)$: an update with $s_i+\check s_i$, after which the matrix is positive definite, and a downdate with $\check s_i$, after which it equals $(n_0+i)S_i$. In the remaining cases $W_i$ is computed by direct inversion; by Lemma~\ref{lem:precondLR}(a) below these cases occur only finitely often. For data that are only mixing we use Bartlett weights $k_h=1-h/(L+1)$ with a bandwidth $L$ \cite{NeweyWest1987}. They can be implemented as $S_i=S_{i-1}+(n_0+i)^{-1}\big(c_ic_i\tr/(L+1)-S_{i-1}\big)$ with the moving sum $c_i=s_i+\dots+s_{i-L}$. This differs from \eqref{eq:lrv} only by edge terms of total size $O(L^2/i)$, is positive definite by construction, and requires a single Sherman--Morrison step.

Let
\[
\kappa_L(a)=\max\Big\{\kappa(a),\ 6,\ \frac{2}{a-1/2}\Big\},
\]
which equals $8$ for $a=3/4$.

\begin{theorem}[$L$-dependent data]\label{thm:dependent}
Replace the independence condition in Assumption~\ref{ass:data} by Assumption~\ref{ass:ldep}, the conditions on $\Theta$, $n_0$, $\beta_0$, $\Phi_0$ and $S_0$ remaining in force. Define $\lambda_0$, $\Lambda$ and $\mathcal W$ with $\Omega_{\rm LR}$ in place of $\Omega$, and use the weights \eqref{eq:lrv}. Under Assumptions~\ref{ass:smooth}--\ref{ass:steps} and $\kappa>2/(2a-1)$, the conclusions of Theorems~\ref{thm:as} and~\ref{thm:rate} hold with $\Omega^*$ replaced by $\Omega^*_{\rm LR}=\Omega_{\rm LR}(\beta^*)$, the latter for every $\varepsilon>\max\{(1-a)/\kappa,\,2/\kappa-a/2,\,3/\kappa-1+a/2\}$. If moreover $\kappa>\kappa_L(a)$, the conclusions of Theorems~\ref{thm:fclt} and~\ref{thm:J} and of Corollary~\ref{cor:inference} hold with $\Omega^*$ replaced by $\Omega^*_{\rm LR}$. In particular $\sqrt n(\bar\beta_n-\beta^*)\Rightarrow N\big(0,(G^{*\top}\Omega_{\rm LR}^{*-1}G^*)^{-1}\big)$, which is the efficient GMM variance under $L$-dependence, $H_n$ estimates it consistently, and $\mathcal J_n\Rightarrow\chi^2_{\dg-\db}$.
\end{theorem}

The proof is in Section~\ref{sec:proofdep}. The difficulty is that $g(X_i,\beta_{i-1})-m(\beta_{i-1})$ is no longer a martingale difference, because $X_i$ depends on $X_{i-1},\dots,X_{i-L}$, which enter $\beta_{i-1}$. We compare the recursion with one whose evaluation point and preconditioner are delayed by $L+1$ steps. The delayed noise $g(X_i,\beta_{i-L-1})-m(\beta_{i-L-1})$ splits into $L+1$ martingale difference sequences, and the cost of the delay is a remainder of order $\gamma_ii^{2/\kappa}+i^{3/\kappa-1}$. Controlling this remainder requires the stronger moment conditions.

The parameter step is thus the same as for independent data. A delay in the evaluation point appears only in the proof, and the price of dependence is a stronger moment condition. The assumption of $L$-dependence with known $L$ is restrictive. For geometrically mixing data, the same algorithm with Bartlett weights and a bandwidth that grows slowly with the sample size is the natural extension, but we have not worked out its theory. Section~\ref{sec:sv} applies it to the stochastic volatility model.

%=====================================================================
\section{Proofs}\label{sec:proofs}
%=====================================================================

Throughout, $C$ denotes a finite deterministic constant whose value may change from line to line; it may depend on the model, on $\Theta$ and on the tuning constants, but not on $i$ or $n$. Finite random constants are written $C(\omega)$. Matrix norms are operator norms. We write $\delta_i=\beta_i-\beta^*$ and $D_\Theta=\operatorname{diam}\Theta$. Recall that $\bar\beta_n=n^{-1}\sum_{i=1}^n\beta_i$, so that every $\bar\beta_i$ lies in the convex set $\Theta$ and
\begin{equation}\label{eq:barstep}
\|\bar\beta_i-\bar\beta_{i-1}\|=\frac{\|\beta_i-\bar\beta_{i-1}\|}{i}\le\frac{D_\Theta}{i}.
\end{equation}

\subsection{Tools and auxiliary lemmas}\label{sec:tools}
We use the following standard results.
\begin{itemize}[leftmargin=*,itemsep=2pt]
\item[(T1)] \emph{Robbins--Siegmund lemma} \cite{RobbinsSiegmund1971}. Let $V_i,b_i,c_i\ge0$ be adapted, with $\E[V_i\mid\F_{i-1}]\le V_{i-1}+b_{i-1}-c_{i-1}$ and $\sum_ib_i<\infty$ almost surely. Then $V_i$ converges almost surely to a finite limit and $\sum_ic_i<\infty$ almost surely.
\item[(T2)] \emph{Martingale strong law.} Let $(\Delta_j)$ be a martingale difference sequence with $\sup_j\E\|\Delta_j\|^2<\infty$. Then $j^{-1/2-\varepsilon'}\sum_{k\le j}\Delta_k\to0$ almost surely for every $\varepsilon'>0$. Indeed, the martingale $\sum_kk^{-1/2-\varepsilon'}\Delta_k$ is bounded in $L^2$ and hence converges almost surely, and Kronecker's lemma applies \cite[Theorem~2.18]{HallHeyde1980}.
\item[(T3)] \emph{Burkholder's inequality} \cite[Theorem~2.10]{HallHeyde1980}, applied coordinatewise. For a martingale difference sequence $(\Delta_k)$ in $\R^{d}$, deterministic weights $w_k$ and $p\ge2$, $\E\|\sum_{k\le t}w_k\Delta_k\|^p\le C_{p,d}\E\big(\sum_{k\le t}w_k^2\|\Delta_k\|^2\big)^{p/2}$. By Jensen's inequality with the probability weights $w_k^2/\sum_lw_l^2$ and the convex function $x\mapsto x^{p/2}$,
\begin{equation}\label{eq:burkholder}
\E\Big\|\sum_{k\le t}w_k\Delta_k\Big\|^p\le C_{p,d}\Big(\sum_{k\le t}w_k^2\Big)^{p/2}\sup_k\E\|\Delta_k\|^p .
\end{equation}
Together with Doob's maximal inequality for the non-negative submartingale $\|\sum_{k\le l}w_k\Delta_k\|$, the left-hand side of \eqref{eq:burkholder} can be replaced by $\E\max_{l\le t}\|\sum_{k\le l}w_k\Delta_k\|^p$, at the price of the factor $(p/(p-1))^p$.
\item[(T4)] \emph{Lenglart's inequality} \cite{Lenglart1977}. Let $M_k$ be a square-integrable martingale with predictable quadratic variation $\langle M\rangle_k$, and let $\eta,\theta>0$. Then $\Prob(\max_{k\le n}\|M_k\|\ge\eta)\le\theta/\eta^2+\Prob(\operatorname{tr}\langle M\rangle_n\ge\theta)$. Hence $\max_{k\le n}\|M_k\|\to0$ in probability whenever $\operatorname{tr}\langle M\rangle_n\to0$ in probability.
\item[(T5)] \emph{Growth of an envelope.} Let $Y_1,Y_2,\dots$ be identically distributed with $\E|Y_1|^\kappa<\infty$. Then $\sum_i\Prob(|Y_i|>i^{1/\kappa})=\sum_i\Prob(|Y_1|^\kappa>i)\le\E|Y_1|^\kappa<\infty$. By the first Borel--Cantelli lemma, almost surely $|Y_i|\le i^{1/\kappa}$ for all $i$ beyond a finite random index. No independence is needed.
\item[(T6)] \emph{Weighted sums.} Let $y_j\ge0$ with $T_i=\sum_{j\le i}y_j\le C(\omega)i$ for all $i$. Then
\[
\sum_{j\le n}y_jj^{-q}\le\frac{C(\omega)}{1-q}\,n^{1-q}\quad(0\le q<1),\qquad\sum_{j\le n}\frac{y_j}{j}\le C(\omega)(2+\log n).
\]
By Abel summation, $\sum_{j\le n}y_jj^{-q}=T_nn^{-q}+\sum_{j<n}T_j\big(j^{-q}-(j+1)^{-q}\big)$, and $0\le j^{-q}-(j+1)^{-q}\le qj^{-q-1}$. The sum is therefore at most $C(\omega)\big(n^{1-q}+q\sum_{j<n}j^{-q}\big)\le C(\omega)n^{1-q}/(1-q)$. The case $q=1$ is identical, with $\sum_{j<n}j^{-1}\le1+\log n$.
\end{itemize}

\begin{lemma}\label{lem:weights}
Let $\gamma_i$ be as in Assumption~\ref{ass:steps} and $\Gamma_{k,i}=\prod_{j=k+1}^i(1-\gamma_j)$, with $\Gamma_{i,i}=1$. There is $i_1$ such that $\gamma_j\le1/2$ for $j\ge i_1$. For $i\ge k\ge i_1$, $0\le\Gamma_{k,i}\le\exp(-\sum_{j=k+1}^i\gamma_j)$ and $\sum_{k=i_1}^i\gamma_k^2\Gamma_{k,i}^2\le C\gamma_i$.
\end{lemma}

\begin{proof}
Since $\gamma_j\to0$, there is $i_1$ with $\gamma_j\le1/2$ for $j\ge i_1$. Every factor of $\Gamma_{k,i}$ then lies in $[1/2,1)$, so $\Gamma_{k,i}\ge0$, and $1-x\le e^{-x}$ gives the first bound. For the second, let $u_i=\sum_{k=i_1}^i\gamma_k^2\Gamma_{k,i}^2$. Since $\Gamma_{k,i}=(1-\gamma_i)\Gamma_{k,i-1}$ for $k<i$ and $\Gamma_{i,i}=1$,
\[
u_i=(1-\gamma_i)^2u_{i-1}+\gamma_i^2\le(1-\gamma_i)u_{i-1}+\gamma_i^2 .
\]
By the mean value theorem, $\gamma_{i-1}-\gamma_i=\gamma_0\big((i+i_0-1)^{-a}-(i+i_0)^{-a}\big)\le a\gamma_{i-1}/(i+i_0-1)$. Since $a<1$, we have $\gamma_{i-1}/\big((i+i_0-1)\gamma_i^2\big)\to0$. Hence $a\gamma_{i-1}/(i+i_0-1)\le\gamma_i^2/2$ for all $i$ larger than some $i_2\ge i_1$. Using $\gamma_i\gamma_{i-1}\ge\gamma_i^2$, this gives
\[
\gamma_{i-1}-\gamma_i-\gamma_i\gamma_{i-1}\le-\gamma_i^2/2\qquad(i>i_2).
\]
Choose $C\ge2$ with $u_i\le C\gamma_i$ for $i_1\le i\le i_2$. If $i>i_2$ and $u_{i-1}\le C\gamma_{i-1}$, then
\[
u_i\le(1-\gamma_i)C\gamma_{i-1}+\gamma_i^2=C\gamma_i+C(\gamma_{i-1}-\gamma_i-\gamma_i\gamma_{i-1})+\gamma_i^2\le C\gamma_i-(C/2-1)\gamma_i^2\le C\gamma_i ,
\]
and the claim follows by induction.
\end{proof}

\begin{lemma}\label{lem:det}
Let $u_i\ge0$ satisfy $u_i\le(1-\gamma_i/2)u_{i-1}+\gamma_i\eta_iv_{i-1}$ for $i>i_3$, where $i_3$ is finite (possibly random), $\eta_i\to0$, and $v_i=i^{-q}$ with $0<q<1$. Then $u_i=o(v_i)$.
\end{lemma}

\begin{proof}
Let $\varsigma_i=u_i/v_i\ge0$. Since $(1+x)^q\le1+qx$ for $x\ge0$ and $i\gamma_i\to\infty$, there is $i_4\ge i_3$ such that for $i>i_4$
\[
\frac{v_{i-1}}{v_i}=\Big(1+\frac1{i-1}\Big)^q\le1+\frac{2q}{i}\le1+\frac{\gamma_i}4\qquad\text{and}\qquad\gamma_i\le1 .
\]
Dividing the recursion by $v_i$ and using $(1-x/2)(1+x/4)=1-x/4-x^2/8$, we get for $i>i_4$
\[
\varsigma_i\le(1-\gamma_i/2)(1+\gamma_i/4)\varsigma_{i-1}+\gamma_i\eta_i(1+\gamma_i/4)\le(1-\gamma_i/4)\varsigma_{i-1}+2\gamma_i|\eta_i| .
\]
Fix $\theta>0$ and $i_5\ge i_4$ with $|\eta_i|\le\theta$ for $i>i_5$. Then $\varsigma_i-8\theta\le(1-\gamma_i/4)(\varsigma_{i-1}-8\theta)$, and hence $(\varsigma_i-8\theta)^+\le(1-\gamma_i/4)(\varsigma_{i-1}-8\theta)^+$, for $i>i_5$. Iterating,
\[
(\varsigma_i-8\theta)^+\le(\varsigma_{i_5}-8\theta)^+\exp\Big(-\frac14\sum_{j=i_5+1}^i\gamma_j\Big)\to0,
\]
because $\sum_j\gamma_j=\infty$. Thus $\limsup_i\varsigma_i\le8\theta$ for every $\theta>0$.
\end{proof}

\begin{lemma}[Jacobian and weight estimates]\label{lem:precond}
Under Assumptions~\ref{ass:data}, \ref{ass:smooth} and~\ref{ass:steps}, almost surely:
\begin{enumerate}[label=(\alph*),itemsep=2pt]
\item $\dist(\Phi_i,\mathcal G)\to0$, and the eigenvalues of $S_i$ lie in $[\lambda_0/2,2\Lambda]$ for all $i$ large, so that $W_i\in\mathcal W$;
\item for every $\varepsilon'>0$,
\[
\Big\|\Phi_i-\frac1i\sum_{j=1}^i\bar G(\bar\beta_{j-1})\Big\|+\Big\|S_i-\frac1i\sum_{j=1}^i\Omega(\bar\beta_{j-1})\Big\|=O\big(i^{-1/2+\varepsilon'}\big);
\]
\item if $\bar\beta_i\to\beta^*$, then $\Phi_i\to G^*$, $S_i\to\Omega^*$, $W_i\to\Omega^{*-1}$, $H_i\to\Sigma^*$ and $P(\Phi_i,W_i)\to P^*$.
\end{enumerate}
\end{lemma}

\begin{proof}
(b) By Remark~\ref{rem:design}(ii),
\[
\Phi_i=\frac{n_0}{n_0+i}\Phi_0+\frac{1}{n_0+i}\sum_{j=1}^i\bar G(\bar\beta_{j-1})+\frac{1}{n_0+i}M^\Phi_i,\qquad M^\Phi_i=\sum_{j=1}^i\Delta^\Phi_j,\quad\Delta^\Phi_j=G(X_j,\bar\beta_{j-1})-\bar G(\bar\beta_{j-1}).
\]
The point $\bar\beta_{j-1}$ is $\F_{j-1}$-measurable, and $X_j$ is independent of $\F_{j-1}$ with the distribution of $X$. Therefore $\E[G(X_j,\bar\beta_{j-1})\mid\F_{j-1}]=\bar G(\bar\beta_{j-1})$, so the $\Delta^\Phi_j$ are martingale differences, with $\E\|\Delta_j^\Phi\|^2\le2\E[B(X)^2]+2\big(\E B(X)\big)^2\le4\E[B(X)^2]$. By (T2), $M^\Phi_i=O(i^{1/2+\varepsilon'})$ almost surely. Since $\frac{1}{n_0+i}-\frac1i=-\frac{n_0}{i(n_0+i)}$ and $\|\bar G(\beta)\|\le\E B(X)$,
\[
\Big\|\Phi_i-\frac1i\sum_{j\le i}\bar G(\bar\beta_{j-1})\Big\|\le\frac{n_0}{n_0+i}\big(\|\Phi_0\|+\E B(X)\big)+\frac{\|M_i^\Phi\|}{n_0+i}=O\big(i^{-1/2+\varepsilon'}\big).
\]
The same argument applies to $S_i=(n_0+i)^{-1}\big(n_0S_0+\sum_{j\le i}s_js_j\tr\big)$, with $\Omega(\bar\beta_{j-1})$ in place of $\bar G(\bar\beta_{j-1})$ and martingale differences $s_js_j\tr-\Omega(\bar\beta_{j-1})$. Their second moments are at most $2\E[B(X)^4]+2\|\Omega(\bar\beta_{j-1})\|^2\le4\E[B(X)^4]<\infty$, since $\kappa\ge4$.

(a) Let $\bar G_i^{\rm av}=i^{-1}\sum_{j\le i}\bar G(\bar\beta_{j-1})$ and $\Omega_i^{\rm av}=i^{-1}\sum_{j\le i}\Omega(\bar\beta_{j-1})$. Since $\bar\beta_{j-1}\in\Theta$, $\bar G_i^{\rm av}$ is a convex combination of elements of $\{\bar G(\beta):\beta\in\Theta\}$, so $\bar G_i^{\rm av}\in\mathcal G$ and $\dist(\Phi_i,\mathcal G)\le\|\Phi_i-\bar G_i^{\rm av}\|\to0$ by (b). The function $\lambda_{\min}$ is concave and $\lambda_{\max}$ is convex on symmetric matrices, so the eigenvalues of $\Omega_i^{\rm av}$ lie in $[\lambda_0,\Lambda]$. By Weyl's inequality, $|\lambda_k(S_i)-\lambda_k(\Omega_i^{\rm av})|\le\|S_i-\Omega_i^{\rm av}\|\to0$ for every $k$. Hence the eigenvalues of $S_i$ lie in $[\lambda_0/2,2\Lambda]$ for $i$ large, and those of $W_i=S_i^{-1}$ lie in $[(2\Lambda)^{-1},2\lambda_0^{-1}]$.

(c) If $\bar\beta_i\to\beta^*$, continuity of $\bar G$ and $\Omega$ gives $\bar G(\bar\beta_{j-1})\to G^*$ and $\Omega(\bar\beta_{j-1})\to\Omega^*$. Hence $\bar G_i^{\rm av}\to G^*$ and $\Omega_i^{\rm av}\to\Omega^*$ by Ces\`aro's lemma, and $\Phi_i\to G^*$, $S_i\to\Omega^*$ by (b). Matrix inversion is continuous on invertible matrices, so $W_i\to\Omega^{*-1}$. The matrix $G^{*\top}\Omega^{*-1}G^*$ is invertible because $G^*$ has full column rank and $\Omega^{*-1}$ is positive definite. Hence $\Phi_i\tr W_i\Phi_i$ is invertible for $i$ large, $H_i\to\Sigma^*$ and $P(\Phi_i,W_i)=H_i\Phi_i\tr W_i\to P^*$.
\end{proof}

We also record a Lipschitz bound for the preconditioner. Let $\mathcal N=\{\Phi:\dist(\Phi,\mathcal G)\le\delta_\Phi\}$, a compact set on which $\sigma_{\min}(\Phi)\ge\sigma_0$ by Assumption~\ref{ass:mono}. Let $\phi_{\max}=\sup_{\mathcal N}\|\Phi\|$. For $\Phi\in\mathcal N$ and $W\in\mathcal W$, $u\tr\Phi\tr W\Phi u\ge\lambda_{\min}(W)\|\Phi u\|^2\ge\sigma_0^2\|u\|^2/(2\Lambda)$. Therefore
\begin{equation}\label{eq:Pbound}
\|(\Phi\tr W\Phi)^{-1}\|\le\frac{2\Lambda}{\sigma_0^2},\qquad\|P(\Phi,W)\|\le p_{\max}:=\frac{4\Lambda\phi_{\max}}{\lambda_0\sigma_0^2}.
\end{equation}
For a second pair $(\Phi',W')\in\mathcal N\times\mathcal W$, write $M=\Phi\tr W\Phi$ and $M'=\Phi'^\top W'\Phi'$. Then $\|M-M'\|\le2\phi_{\max}\|W\|\,\|\Phi-\Phi'\|+\phi_{\max}^2\|W-W'\|$ and $\|M^{-1}-M'^{-1}\|\le\|M^{-1}\|\,\|M-M'\|\,\|M'^{-1}\|$. Expanding $M^{-1}\Phi\tr W-M'^{-1}\Phi'^\top W'$ term by term, we obtain, for some constant $C_P$,
\begin{equation}\label{eq:Plip}
\|P(\Phi,W)-P(\Phi',W')\|\le C_P\big(\|\Phi-\Phi'\|+\|W-W'\|\big)\qquad\text{on }\mathcal N\times\mathcal W .
\end{equation}

\subsection{Proofs of the results of Section~\ref{sec:theory}}

\begin{proof}[Proof of Proposition~\ref{prop:local}]
(a) For a compact convex $\Theta\subseteq\Theta_0$, let $\lambda_0(\Theta)$, $\Lambda(\Theta)$, $\mathcal W_\Theta$ and $\mathcal G_\Theta$ be the quantities of Section~\ref{sec:theory} computed on $\Theta$. Since $\lambda_0(\Theta)\ge\lambda_0(\Theta_0)$ and $\Lambda(\Theta)\le\Lambda(\Theta_0)$, we have $\mathcal W_\Theta\subseteq\mathcal W_{\Theta_0}$. Every $W\in\mathcal W_{\Theta_0}$ has condition number at most $k_W=4\Lambda(\Theta_0)/\lambda_0(\Theta_0)$. Let $\sigma^*=\sigma_{\min}(G^*)>0$. Take $\delta_\Phi>0$ and $r_1\in(0,r_0]$ with $\delta_\Phi+\bar Lr_1\le\sigma^*/2$, and let $\Theta\subseteq\bar B(\beta^*,r_1)$ with $\beta^*\in\operatorname{int}\Theta$.

Since $\bar G$ is $\bar L$-Lipschitz, $\mathcal G_\Theta\subseteq\bar B(G^*,\bar Lr_1)$. So if $\dist(\Phi,\mathcal G_\Theta)\le\delta_\Phi$, then $\|\Phi-G^*\|\le\delta_\Phi+\bar Lr_1\le\sigma^*/2$. By Weyl's inequality for singular values, $\sigma_{\min}(\Phi)\ge\sigma^*/2=:\sigma_0$, and $\|\Phi\|\le\|G^*\|+\sigma_0$. As in \eqref{eq:Pbound}, $\|(\Phi\tr W\Phi)^{-1}\|\le1/(\lambda_{\min}(W)\sigma_{\min}(\Phi)^2)$, so for $W\in\mathcal W_\Theta$
\[
\|P(\Phi,W)\|\le\frac{\lambda_{\max}(W)\|\Phi\|}{\lambda_{\min}(W)\sigma_{\min}(\Phi)^2}\le\frac{k_W(\|G^*\|+\sigma_0)}{\sigma_0^2}=:p_0 .
\]
For $\beta\in\Theta$, the segment from $\beta^*$ to $\beta$ lies in $\Theta\subseteq\mathcal O$. Hence $g(x,\beta)-g(x,\beta^*)=\int_0^1G(x,\beta^*+t(\beta-\beta^*))\dd t\,(\beta-\beta^*)$. Since $\|G\|\le B$ is integrable and $m(\beta^*)=0$, Fubini's theorem gives
\begin{gather*}
m(\beta)=\tilde G(\beta)(\beta-\beta^*),\qquad\tilde G(\beta)=\int_0^1\bar G\big(\beta^*+t(\beta-\beta^*)\big)\dd t,\\ \|\tilde G(\beta)-G^*\|\le\int_0^1\bar Lt\|\beta-\beta^*\|\dd t\le\frac{\bar Lr_1}2 .
\end{gather*}
Using $P(\Phi,W)\Phi=I$,
\[
P(\Phi,W)m(\beta)=(\beta-\beta^*)+P(\Phi,W)\big(\tilde G(\beta)-\Phi\big)(\beta-\beta^*),\qquad\|\tilde G(\beta)-\Phi\|\le\frac{\bar Lr_1}2+\delta_\Phi+\bar Lr_1 .
\]
Hence $\langle\beta-\beta^*,P(\Phi,W)m(\beta)\rangle\ge\big(1-p_0(\delta_\Phi+2\bar Lr_1)\big)\|\beta-\beta^*\|^2$. Choosing $\delta_\Phi$ and $r_1$ smaller if necessary, so that $p_0(\delta_\Phi+2\bar Lr_1)\le1/2$, gives the claim with $c=1/2$. Here $p_0$ depends only on $\Theta_0$ and $\sigma_0=\sigma^*/2$, so the constants $\delta_\Phi$, $\sigma_0$ and $c$ do not depend on $\Theta$.

(b) In the linear model, $\bar G\equiv G^*$, $\mathcal G=\{G^*\}$ and $m(\beta)=G^*(\beta-\beta^*)$ on all of $\R^{\db}$. The argument of (a) applies with $\bar L=0$, $\Theta_0=\Theta$ and no restriction on $r_1$.
\end{proof}

\begin{proof}[Proof of Corollary~\ref{cor:datadriven}]
On $\mathcal E_\rho$, every $\beta\in\Theta$ satisfies $\|\beta-\beta^*\|\le\|\beta-\hat\beta_{\rm pilot}\|+\|\hat\beta_{\rm pilot}-\beta^*\|<2\rho\le r_1$. The open ball around $\beta^*$ of radius $\rho-\|\hat\beta_{\rm pilot}-\beta^*\|>0$ is contained in $\Theta$, so $\beta^*\in\operatorname{int}\Theta$.

Let $\mathcal P$ denote the pilot sample. All inputs of Algorithm~\ref{alg:sgmm} that depend on $\mathcal P$ are measurable functions of it: $\Theta$, $\beta_0$, $\Phi_0$ and $S_0$. The iterates are therefore measurable functions of $(\mathcal P,X_1,X_2,\dots)$. Fix a value $p$ of $\mathcal P$ with $\|\hat\beta_{\rm pilot}(p)-\beta^*\|<\rho$, and run the algorithm on the stream $(X_i)_{i\ge1}$ with the deterministic inputs $\Theta(p)$, $\beta_0(p)$, $\Phi_0(p)$, $S_0(p)$. Assumption~\ref{ass:data} holds. Assumption~\ref{ass:smooth} holds on $\Theta(p)\subseteq\Theta_0$, because the envelope $B$ and the Lipschitz function $L$ computed over $\Theta(p)$ are bounded by those computed over $\Theta_0$, and $\lambda_0(\Theta(p))\ge\lambda_0(\Theta_0)>0$. Assumption~\ref{ass:mono} holds by Proposition~\ref{prop:local}(a), and Assumption~\ref{ass:steps} does not involve $p$. Hence the conclusions of Theorems~\ref{thm:as}--\ref{thm:J} hold for this deterministic input, under the respective moment conditions.

Since $(X_i)_{i\ge1}$ is independent of $\mathcal P$, Fubini's theorem gives, for every measurable event $A$ of the form $\{(\mathcal P,(X_i)_{i\ge1})\in\mathcal A\}$, $\Prob(A)=\int\Prob\big((p,(X_i)_{i\ge1})\in\mathcal A\big)\,\Prob_{\mathcal P}(\dd p)$. With $A=\{\bar\beta_n\to\beta^*\}$, the integrand equals one whenever $\|\hat\beta_{\rm pilot}(p)-\beta^*\|<\rho$, so $\Prob(\bar\beta_n\to\beta^*)\ge\Prob(\mathcal E_\rho)$. The distributional conclusions hold in the same way under the conditional distribution given $\mathcal P=p$, for such $p$. If the pilot estimator is consistent as $n_0\to\infty$, then $\Prob(\mathcal E_\rho)\to1$.
\end{proof}

\begin{proof}[Proof of Theorem~\ref{thm:as}]
\emph{Step 1: eventual regime.} Define the $\F_{i-1}$-measurable event
\[
E_{i-1}=\big\{\Phi_{i-1}\in\mathcal N,\ W_{i-1}\in\mathcal W,\ P_{i-1}=P(\Phi_{i-1},W_{i-1})\big\}.
\]
By Lemma~\ref{lem:precond}(a), almost surely $\Phi_{i-1}\in\mathcal N$ and $W_{i-1}\in\mathcal W$ for all $i$ large. Then $\Phi_{i-1}\tr W_{i-1}\Phi_{i-1}$ is invertible and $\|P(\Phi_{i-1},W_{i-1})\|\le p_{\max}$ by \eqref{eq:Pbound}. Since $\bar p_i\to\infty$, $p_{\max}\le\bar p_i$ for $i$ large, so the truncation is inactive. Hence, almost surely, $E_{i-1}$ holds for all $i$ beyond a finite random time.

\emph{Step 2: Robbins--Siegmund.} Let $V_i=\|\delta_i\|^2$. Since $\Pi_\Theta$ is non-expansive and $\beta^*=\Pi_\Theta(\beta^*)$,
\[
V_i\le\|\delta_{i-1}-\gamma_iP_{i-1}g(X_i,\beta_{i-1})\|^2=V_{i-1}-2\gamma_i\langle\delta_{i-1},P_{i-1}g(X_i,\beta_{i-1})\rangle+\gamma_i^2\|P_{i-1}g(X_i,\beta_{i-1})\|^2 .
\]
The variables $\delta_{i-1}$, $\beta_{i-1}$ and $P_{i-1}$ are $\F_{i-1}$-measurable, and $X_i$ is independent of $\F_{i-1}$. Hence $\E[g(X_i,\beta_{i-1})\mid\F_{i-1}]=m(\beta_{i-1})$ and $\E[\|g(X_i,\beta_{i-1})\|^2\mid\F_{i-1}]\le\E[B(X)^2]$. By construction $P_{i-1}$ is either $\hat P$ with $\|\hat P\|\le\bar p_i$ or zero, so $\|P_{i-1}\|\le\bar p_i$. Therefore
\[
\E[V_i\mid\F_{i-1}]\le V_{i-1}-2\gamma_i\langle\delta_{i-1},P_{i-1}m(\beta_{i-1})\rangle+\gamma_i^2\bar p_i^2\E[B(X)^2].
\]
On $E_{i-1}$, the inner product is at least $cV_{i-1}$ by Assumption~\ref{ass:mono}. On its complement it is at least $-\bar p_iK_0$, with $K_0=D_\Theta\sup_\Theta\|m\|\le D_\Theta\E B(X)$. Therefore
\[
\E[V_i\mid\F_{i-1}]\le V_{i-1}-\underbrace{2c\gamma_iV_{i-1}\mathbf 1_{E_{i-1}}}_{c_{i-1}}+\underbrace{2\gamma_i\bar p_iK_0\mathbf 1_{E_{i-1}^c}+\gamma_i^2\bar p_i^2\E[B(X)^2]}_{b_{i-1}} .
\]
Both $b_{i-1}$ and $c_{i-1}$ are non-negative and $\F_{i-1}$-measurable. The first part of $b_{i-1}$ vanishes for all $i$ large by Step~1. The second is deterministic and summable, because $\gamma_i^2\bar p_i^2\le Ci^{-2a+2\nu}$ and $2a-2\nu>1$. By (T1), $V_i$ converges almost surely to a finite limit $V_\infty$, and $\sum_i\gamma_iV_{i-1}\mathbf 1_{E_{i-1}}<\infty$. On the event $\{V_\infty>0\}$ we would have $V_{i-1}\mathbf 1_{E_{i-1}}\ge V_\infty/2$ for all $i$ large, which contradicts $\sum_i\gamma_i=\infty$. Hence $V_\infty=0$ almost surely, that is, $\beta_i\to\beta^*$, and $\bar\beta_n=n^{-1}\sum_{i\le n}\beta_i\to\beta^*$ by Ces\`aro's lemma.

\emph{Step 3: limits of the preconditioner and of the projection.} Lemma~\ref{lem:precond}(c) gives the limits of $\Phi_n$, $W_n$, $H_n$ and $P(\Phi_n,W_n)$, and $P_n=P(\Phi_n,W_n)$ for $n$ large by Step~1. For the projection, let $r>0$ with $\bar B(\beta^*,2r)\subseteq\Theta$. By (T5), $B(X_i)\le i^{1/\kappa}$ for all $i$ large. On $E_{i-1}$,
\[
\gamma_i\|P_{i-1}g(X_i,\beta_{i-1})\|\le\gamma_0p_{\max}i^{-a+1/\kappa}\to0,
\]
because $a>1/2\ge2/\kappa>1/\kappa$. Since also $\|\delta_{i-1}\|\to0$, the point $\beta_{i-1}-\gamma_iP_{i-1}g(X_i,\beta_{i-1})$ lies in $\bar B(\beta^*,2r)\subseteq\Theta$ for all $i$ large, and the projection is inactive.
\end{proof}

\begin{proof}[Proof of Theorem~\ref{thm:rate}]
Fix $\varepsilon>(1-a)/\kappa$; without loss of generality $\varepsilon<a/2$, which is possible because $(1-a)/\kappa\le(1-a)/4<a/2$. Let $v_i=i^{-a/2+\varepsilon}$. Let $\iota$ be a finite random time after which the truncation and the projection are inactive and $E_{i-1}$ holds (Theorem~\ref{thm:as}). Let $e_i=g(X_i,\beta_{i-1})-m(\beta_{i-1})$. It is a martingale difference with $\E[\|e_i\|^\kappa\mid\F_{i-1}]\le2^\kappa\E[B(X)^\kappa]$. With $\tilde G$ as in the proof of Proposition~\ref{prop:local}, $m(\beta_{i-1})=\tilde G(\beta_{i-1})\delta_{i-1}$; the representation used there only requires the segment from $\beta^*$ to $\beta_{i-1}$ to lie in the convex set $\Theta$. For $i>\iota$,
\begin{equation}\label{eq:linrec}
\delta_i=(I-\gamma_iA_{i-1})\delta_{i-1}-\gamma_iP_{i-1}e_i,\qquad A_{i-1}=P_{i-1}\tilde G(\beta_{i-1}).
\end{equation}
Since $\beta_{i-1}\to\beta^*$, $\tilde G(\beta_{i-1})\to G^*$ by continuity of $\bar G$. With $P_{i-1}\to P^*$ (Theorem~\ref{thm:as}), this gives $A_{i-1}\to P^*G^*=I$ almost surely.

\emph{Step 1: a scalar-weighted martingale.} Let $\tilde P_{i-1}=P_{i-1}\mathbf 1\{\|P_{i-1}\|\le p_{\max}\}$, which is $\F_{i-1}$-measurable and bounded. Then $\xi_i=\tilde P_{i-1}e_i$ is a martingale difference with $\E\|\xi_i\|^\kappa\le c_\kappa:=p_{\max}^\kappa2^\kappa\E[B(X)^\kappa]$. With $i_1$ from Lemma~\ref{lem:weights}, define $Z_{i_1}=0$ and $Z_i=(1-\gamma_i)Z_{i-1}-\gamma_i\xi_i$ for $i>i_1$. Dividing by $\Gamma_{i_1,i}=(1-\gamma_i)\Gamma_{i_1,i-1}$ and summing, we get $Z_i=-\Gamma_{i_1,i}M_i$ with
\[
M_i=\sum_{k=i_1+1}^i\frac{\gamma_k}{\Gamma_{i_1,k}}\xi_k ,
\]
a martingale with deterministic weights.

Define block endpoints $t_0=i_1$ and $t_{k+1}=\min\{t>t_k:\sum_{j=t_k+1}^t\gamma_j\ge1\}$; they are finite because $\sum_j\gamma_j=\infty$. Since $\gamma_j\le1/2$, each block satisfies $1\le\sum_{j=t_k+1}^{t_{k+1}}\gamma_j<3/2$. Consequently $k\le\sum_{j\le t_k}\gamma_j\le Ct_k^{1-a}$.

By (T3) with $p=\kappa$, in the maximal form, and since $\Gamma_{i_1,t}=\Gamma_{i_1,j}\Gamma_{j,t}$ for $i_1\le j\le t$, Lemma~\ref{lem:weights} gives
\[
\E\max_{l\le t}\|M_l\|^\kappa\le Cc_\kappa\Big(\sum_{j=i_1+1}^t\frac{\gamma_j^2}{\Gamma_{i_1,j}^2}\Big)^{\kappa/2}=C\,\Gamma_{i_1,t}^{-\kappa}\Big(\sum_{j=i_1+1}^t\gamma_j^2\Gamma_{j,t}^2\Big)^{\kappa/2}\le C\,\Gamma_{i_1,t}^{-\kappa}\gamma_t^{\kappa/2}.
\]
For $t_k<i\le t_{k+1}$, $\Gamma_{i_1,i}\le\Gamma_{i_1,t_k}$, so $\|Z_i\|\le\Gamma_{i_1,t_k}\max_{l\le t_{k+1}}\|M_l\|$. Since $1-x\ge e^{-x-x^2}$ for $x\in[0,1/2]$, and $\sum_{j=t_k+1}^{t_{k+1}}\gamma_j^2\le\frac12\sum_{j=t_k+1}^{t_{k+1}}\gamma_j<3/4$, we have $\Gamma_{t_k,t_{k+1}}\ge e^{-3/2-3/4}=e^{-9/4}$. Therefore
\[
\E\max_{t_k<i\le t_{k+1}}\|Z_i\|^\kappa\le C\Big(\frac{\Gamma_{i_1,t_k}}{\Gamma_{i_1,t_{k+1}}}\Big)^\kappa\gamma_{t_{k+1}}^{\kappa/2}=C\,\Gamma_{t_k,t_{k+1}}^{-\kappa}\gamma_{t_{k+1}}^{\kappa/2}\le C\gamma_{t_{k+1}}^{\kappa/2}.
\]
Choose $\theta>0$ with $(1+\theta)(1-a)/\kappa<\varepsilon$. By Markov's inequality,
\[
\Prob\Big(\max_{t_k<i\le t_{k+1}}\|Z_i\|>k^{(1+\theta)/\kappa}\gamma_{t_{k+1}}^{1/2}\Big)\le\frac{C\gamma_{t_{k+1}}^{\kappa/2}}{k^{1+\theta}\gamma_{t_{k+1}}^{\kappa/2}}=Ck^{-1-\theta},
\]
which is summable in $k$. By Borel--Cantelli, almost surely, for all large $k$ and all $t_k<i\le t_{k+1}$,
\[
\|Z_i\|\le k^{(1+\theta)/\kappa}\gamma_{t_{k+1}}^{1/2}\le Ci^{(1+\theta)(1-a)/\kappa}\gamma_i^{1/2}\le Cv_i,
\]
where we used $k\le Ct_k^{1-a}\le Ci^{1-a}$ and $\gamma_{t_{k+1}}\le\gamma_i$. Hence $\|Z_i\|=O(v_i)$ almost surely.

\emph{Step 2: comparison.} Let $\iota'\ge\max(\iota,i_1)$ be a finite random time after which $\|P_{i-1}\|\le p_{\max}$, so that $\xi_i=P_{i-1}e_i$, and $\|A_{i-1}-I\|\le1/2$. For $i>\iota'$, \eqref{eq:linrec} and the definition of $Z_i$ give for $\Delta_i=\delta_i-Z_i$
\[
\Delta_i=(I-\gamma_iA_{i-1})\delta_{i-1}-(1-\gamma_i)Z_{i-1}=(I-\gamma_iA_{i-1})\Delta_{i-1}-\gamma_i(A_{i-1}-I)Z_{i-1}.
\]
Since $I-\gamma_iA_{i-1}=(1-\gamma_i)I-\gamma_i(A_{i-1}-I)$, we have $\|I-\gamma_iA_{i-1}\|\le1-\gamma_i/2$. Hence
\[
\|\Delta_i\|\le(1-\gamma_i/2)\|\Delta_{i-1}\|+\gamma_i\eta_iv_{i-1},\qquad\eta_i=\frac{\|A_{i-1}-I\|\,\|Z_{i-1}\|}{v_{i-1}}\to0 .
\]
Lemma~\ref{lem:det} with $q=a/2-\varepsilon\in(0,1)$ gives $\|\Delta_i\|=o(v_i)$, hence $\|\delta_i\|=O(v_i)$. For the average, $\|\bar\beta_n-\beta^*\|\le n^{-1}\sum_{i\le n}\|\delta_i\|\le C(\omega)n^{-1}\sum_{i\le n}i^{-q}\le C(\omega)v_n/(1-q)$.

\emph{Step 3: Jacobian and weights.} By Lemma~\ref{lem:precond}(b) with $\varepsilon'<\varepsilon$ and the Lipschitz continuity of $\bar G$,
\[
\|\Phi_i-G^*\|\le\frac{\bar L}i\sum_{j\le i}\|\bar\beta_{j-1}-\beta^*\|+O\big(i^{-1/2+\varepsilon'}\big)=O(v_i),
\]
because $i^{-1}\sum_{j\le i}v_j=O(v_i)$ and $-1/2+\varepsilon'<-a/2+\varepsilon$. The same argument applies to $S_i-\Omega^*$, using the Lipschitz continuity of $\Omega$. Finally, $\|W_i-\Omega^{*-1}\|=\|W_i(\Omega^*-S_i)\Omega^{*-1}\|\le\|W_i\|\,\|S_i-\Omega^*\|\,\|\Omega^{*-1}\|$, and $\|W_i\|\le2/\lambda_0$ for $i$ large.
\end{proof}

\begin{proof}[Proof of Theorem~\ref{thm:fclt}]
Since $\kappa>\kappa(a)$, we have $(1-a)/\kappa<(1-a)/2$ and $(1-a)/\kappa<(2a-1)/4$. We can therefore fix $\varepsilon$ with
\begin{equation}\label{eq:eps}
\frac{1-a}\kappa<\varepsilon<\min\Big\{\frac{1-a}2,\ \frac{2a-1}4\Big\},
\end{equation}
and Theorem~\ref{thm:rate} applies with $v_i=i^{-a/2+\varepsilon}$. Let $\iota$ be as in the previous proof. For $i>\iota$, \eqref{eq:linrec} gives $\delta_i=\delta_{i-1}-\gamma_iP_{i-1}m(\beta_{i-1})-\gamma_iP_{i-1}e_i$, which can be rearranged as
\[
\delta_{i-1}=\frac{\delta_{i-1}-\delta_i}{\gamma_i}-P_{i-1}e_i-\rho_{i-1},\qquad\rho_{i-1}=P_{i-1}m(\beta_{i-1})-\delta_{i-1}.
\]
Write $e_i=\xi^0_i+\zeta_i$, where $\xi^0_i=g(X_i,\beta^*)$ are i.i.d.\ with mean $m(\beta^*)=0$ and covariance $\Omega^*$, and
\[
\zeta_i=g(X_i,\beta_{i-1})-g(X_i,\beta^*)-m(\beta_{i-1}).
\]
The $\zeta_i$ are martingale differences, and since $\|g(x,\beta)-g(x,\beta^*)\|\le B(x)\|\beta-\beta^*\|$,
\[
\E[\|\zeta_i\|^2\mid\F_{i-1}]\le\E[B(X)^2]\,\|\delta_{i-1}\|^2 .
\]
Using $P_{i-1}e_i=P^*\xi^0_i+P^*\zeta_i+(P_{i-1}-P^*)e_i$, we obtain for $\iota<k\le n$
\begin{equation}\label{eq:decomp}
\sum_{i=\iota+1}^k\delta_{i-1}=-P^*\sum_{i=\iota+1}^k\xi_i^0+R^{(1)}_k-R^{(2)}_k-R^{(3)}_k-R^{(4)}_k,
\end{equation}
where
\[
\textstyle R^{(1)}_k=\sum_{i=\iota+1}^k\frac{\delta_{i-1}-\delta_i}{\gamma_i},\quad R^{(2)}_k=P^*\sum_{i=\iota+1}^k\zeta_i,\quad R^{(3)}_k=\sum_{i=\iota+1}^k(P_{i-1}-P^*)e_i,\quad R^{(4)}_k=\sum_{i=\iota+1}^k\rho_{i-1}.
\]
Here $\iota$ is a random time, finite almost surely but not a stopping time. We show that $n^{-1/2}\max_{\iota<k\le n}\|R^{(l)}_k\|\to0$ in probability for $l=1,\dots,4$. The terms with $i\le\iota$ are finitely many and do not depend on $n$: $n^{-1/2}\sum_{i\le\iota}\|\delta_{i-1}\|\le n^{-1/2}\iota D_\Theta\to0$ and $n^{-1/2}\|P^*\sum_{i\le\iota}\xi^0_i\|\to0$ almost surely. Let $\iota'\ge\iota$ be as in the proof of Theorem~\ref{thm:rate}.

\emph{$R^{(1)}$.} By Abel summation,
\[
R^{(1)}_k=\frac{\delta_\iota}{\gamma_{\iota+1}}-\frac{\delta_k}{\gamma_k}+\sum_{i=\iota+1}^{k-1}\delta_i\Big(\frac1{\gamma_{i+1}}-\frac1{\gamma_i}\Big),
\]
and $0\le1/\gamma_{i+1}-1/\gamma_i=\gamma_0^{-1}\big((i+1+i_0)^a-(i+i_0)^a\big)\le a\gamma_0^{-1}(i+i_0)^{a-1}$. With $\|\delta_i\|\le C(\omega)v_i$ for all $i$,
\[
\max_{\iota<k\le n}\|R^{(1)}_k\|\le C(\omega)\Big(1+\max_{k\le n}k^av_k+\sum_{i\le n}i^{a-1}v_i\Big)\le C(\omega)\big(1+n^{a/2+\varepsilon}\big),
\]
and $n^{a/2+\varepsilon-1/2}\to0$ by \eqref{eq:eps}.

\emph{$R^{(4)}$.} Since $\|G(x,\beta)-G(x,\beta')\|\le L(x)\|\beta-\beta'\|$, the argument for $\tilde G$ gives $\|m(\beta)-G^*(\beta-\beta^*)\|\le\bar L\|\beta-\beta^*\|^2/2$. Using also $P_{i-1}\Phi_{i-1}=I$, for $i>\iota$
\[
\begin{aligned}\|\rho_{i-1}\|&=\big\|P_{i-1}\big(m(\beta_{i-1})-G^*\delta_{i-1}\big)+P_{i-1}(G^*-\Phi_{i-1})\delta_{i-1}\big\|\\&\le p_{\max}\big(\bar L\|\delta_{i-1}\|^2/2+\|\Phi_{i-1}-G^*\|\,\|\delta_{i-1}\|\big)\le C(\omega)v_{i-1}^2\end{aligned}
\]
by Theorem~\ref{thm:rate}. Thus $\max_{k\le n}\|R^{(4)}_k\|\le C(\omega)\sum_{i\le n}i^{-a+2\varepsilon}\le C(\omega)n^{1-a+2\varepsilon}$, where $a-2\varepsilon<1$. Moreover $n^{1/2-a+2\varepsilon}\to0$ by \eqref{eq:eps}.

\emph{$R^{(2)}$.} Since $\iota$ is not a stopping time, we compare with $M^{(2)}_k=P^*\sum_{i=1}^k\zeta_i$, a martingale for $(\F_k)$. Since $R^{(2)}_k=M^{(2)}_k-M^{(2)}_{\iota\wedge k}$, we have $\max_{k\le n}\|R^{(2)}_k\|\le2\max_{k\le n}\|M^{(2)}_k\|$. Its predictable quadratic variation satisfies
\[
\operatorname{tr}\langle M^{(2)}\rangle_n\le\|P^*\|^2\E[B(X)^2]\sum_{i\le n}\|\delta_{i-1}\|^2,
\]
and $n^{-1}\sum_{i\le n}\|\delta_{i-1}\|^2\to0$ almost surely by Theorem~\ref{thm:as} and Ces\`aro's lemma. By (T4) applied to $n^{-1/2}M^{(2)}$, $n^{-1/2}\max_{k\le n}\|R^{(2)}_k\|\to0$ in probability.

\emph{$R^{(3)}$.} Let $M^{(3)}_k=\sum_{i=1}^k(\tilde P_{i-1}-P^*)e_i$, with $\tilde P_{i-1}$ the bounded predictable matrix of the previous proof. It is a martingale with
\[
\operatorname{tr}\langle M^{(3)}\rangle_n\le4\E[B(X)^2]\sum_{i\le n}\|\tilde P_{i-1}-P^*\|^2,
\]
and $n^{-1}\sum_{i\le n}\|\tilde P_{i-1}-P^*\|^2\to0$ almost surely, because $\tilde P_{i-1}=P_{i-1}\to P^*$. By (T4), $n^{-1/2}\max_{k\le n}\|M^{(3)}_k\|\to0$ in probability. Since $\tilde P_{i-1}=P_{i-1}$ for $i>\iota'$, $R^{(3)}_k$ differs from $M^{(3)}_k-M^{(3)}_{\iota'\wedge k}$ by a sum of finitely many terms that does not depend on $n$, and the claim follows.

\emph{Conclusion.} Since $\sum_{i=1}^{\lfloor nt\rfloor}(\beta_i-\beta^*)=\sum_{i=1}^{\lfloor nt\rfloor}\delta_{i-1}+\delta_{\lfloor nt\rfloor}-\delta_0$ and $\max_k\|\delta_k\|\le D_\Theta$, \eqref{eq:decomp} and the bounds above give \eqref{eq:bahadur}. By Donsker's theorem \cite[Theorem~14.1]{Billingsley1999}, applied to $\lambda\tr\xi^0_i$ for each $\lambda\in\R^{\dg}$, together with the Cram\'er--Wold device for the finite-dimensional distributions and the fact that a vector of tight coordinate processes with continuous limits is tight, we have $n^{-1/2}\sum_{i\le\lfloor nt\rfloor}\xi^0_i\Rightarrow\Omega^{*1/2}\mathbf B_g$ in $D[0,1]^{\dg}$, with $\mathbf B_g$ a standard $\dg$-dimensional Brownian motion. Moreover
\[
P^*\Omega^*P^{*\top}=\Sigma^*G^{*\top}\Omega^{*-1}\Omega^*\Omega^{*-1}G^*\Sigma^*=\Sigma^* .
\]
Since the limit has continuous paths, \eqref{eq:bahadur} and Slutsky's lemma give $\mathbb B_n\Rightarrow-P^*\Omega^{*1/2}\mathbf B_g$, a Gaussian process with the law of $\Sigma^{*1/2}\mathbf B$. Finally, $\sqrt n(\bar\beta_n-\beta^*)=\mathbb B_n(1)$.
\end{proof}

\begin{proof}[Proof of Corollary~\ref{cor:inference}]
(a) By Theorem~\ref{thm:as}, $H_{n,kk}\to\Sigma^*_{kk}>0$ almost surely. By Theorem~\ref{thm:fclt} at $t=1$, $\sqrt n(\bar\beta_{n,k}-\beta^*_k)\Rightarrow N(0,\Sigma^*_{kk})$. Slutsky's lemma gives $\sqrt n(\bar\beta_{n,k}-\beta^*_k)/H_{n,kk}^{1/2}\Rightarrow N(0,1)$, so the coverage probability converges to $\Prob(|N(0,1)|\le z_{1-\alpha/2})=1-\alpha$.

(b) Since $s(\bar\beta_s-\beta^*)=\sqrt n\,\mathbb B_n(s/n)$, we have $s(\bar\beta_s-\bar\beta_n)=\sqrt n\big(\mathbb B_n(s/n)-\frac sn\mathbb B_n(1)\big)$ and
\[
V_n=\frac1n\sum_{s=1}^n\Big(\mathbb B_n(\tfrac sn)-\tfrac sn\mathbb B_n(1)\Big)\Big(\mathbb B_n(\tfrac sn)-\tfrac sn\mathbb B_n(1)\Big)\tr .
\]
Consider the map $\Psi$ from $D[0,1]^{\db}$ to $\R^{\db}\times\R^{\db\times\db}$ that sends $f$ to $\big(f(1),\int_0^1(f(r)-rf(1))(f(r)-rf(1))\tr\dd r\big)$. Convergence in the Skorokhod topology to a continuous function is uniform convergence, and $\Psi$ is continuous for the uniform norm. Hence $\Psi$ is continuous at every continuous $f$.

The process $\mathbb B_n$ is constant on each interval $[(s-1)/n,s/n)$, and $\|\mathbb B_n(s/n)-\mathbb B_n((s-1)/n)\|=n^{-1/2}\|\delta_s\|\le n^{-1/2}D_\Theta$. Also $|r-s/n|\le1/n$ on that interval. Hence each summand of $V_n$ differs from the integral of the corresponding integrand over $[(s-1)/n,s/n)$ by at most $Cn^{-1}\cdot n^{-1/2}\big(1+\sup_t\|\mathbb B_n(t)\|\big)^2$. The Riemann sum $V_n$ therefore differs from the second component of $\Psi(\mathbb B_n)$ by at most $Cn^{-1/2}\big(1+\sup_t\|\mathbb B_n(t)\|\big)^2$, which tends to zero in probability because $\sup_t\|\mathbb B_n(t)\|$ is tight. By Theorem~\ref{thm:fclt} and the continuous mapping theorem,
\[
\big(\sqrt n(\bar\beta_n-\beta^*),V_n\big)\Rightarrow\Big(\Sigma^{*1/2}\mathbf B(1),\ \Sigma^{*1/2}\int_0^1(\mathbf B(r)-r\mathbf B(1))(\mathbf B(r)-r\mathbf B(1))\tr\dd r\,\Sigma^{*1/2}\Big).
\]
The $k$th coordinate of $\Sigma^{*1/2}\mathbf B$ is $(\Sigma^*_{kk})^{1/2}$ times a standard Brownian motion $\mathbf b$, and the $(k,k)$ entry of the second component is $\Sigma^*_{kk}\int_0^1(\mathbf b(r)-r\mathbf b(1))^2\dd r$. The factor $\Sigma^*_{kk}$ cancels in the ratio. The denominator of the limit is almost surely positive, because a Brownian bridge is not identically zero. The continuous mapping theorem for the ratio gives the stated limit.
\end{proof}

\begin{proof}[Proof of Theorem~\ref{thm:J}]
Let $\varepsilon$ be as in \eqref{eq:eps} and $v_i=i^{-a/2+\varepsilon}$.

\emph{Step 1: linearization error.} For $b,b'\in\Theta$, Taylor's formula with integral remainder and the Lipschitz bound on $G$ give
\[
\big\|g(x,b)-g(x,b')-G(x,b')(b-b')\big\|=\Big\|\int_0^1\big(G(x,b'+t(b-b'))-G(x,b')\big)\dd t\,(b-b')\Big\|\le\tfrac12L(x)\|b-b'\|^2 .
\]
With $b=\bar\beta_n$, $b'=\bar\beta_{i-1}$ and $\|b-b'\|^2\le2\|\bar\beta_n-\beta^*\|^2+2\|\bar\beta_{i-1}-\beta^*\|^2$, summing over $i\le n$ gives
\[
\sqrt n\|\hat g_n-\bar g_n(\bar\beta_n)\|\le\sqrt n\,\|\bar\beta_n-\beta^*\|^2\,\frac1n\sum_{i\le n}L(X_i)+\frac{1}{\sqrt n}\sum_{i\le n}L(X_i)\|\bar\beta_{i-1}-\beta^*\|^2 .
\]
By Theorem~\ref{thm:rate} and the strong law of large numbers for $L(X_i)$, the first term is $O(n^{1/2}v_n^2)=O(n^{1/2-a+2\varepsilon})$. For the second, $\|\bar\beta_{i-1}-\beta^*\|^2\le C(\omega)i^{-a+2\varepsilon}$ for $i\ge1$. Moreover $T_i=\sum_{j\le i}L(X_j)\le C(\omega)i$ by the strong law, and $q=a-2\varepsilon\in(0,1)$ by \eqref{eq:eps}. By (T6), $\sum_{i\le n}L(X_i)i^{-a+2\varepsilon}\le C(\omega)n^{1-a+2\varepsilon}$, so the second term is also $O(n^{1/2-a+2\varepsilon})$. Both tend to zero by \eqref{eq:eps}, because $1/2-a+2\varepsilon<0$ is equivalent to $\varepsilon<(2a-1)/4$. This proves the first claim.

\emph{Step 2: limit.} By the integral form of the mean value theorem,
\[
\bar g_n(\bar\beta_n)=\bar g_n(\beta^*)+\check G_n(\bar\beta_n-\beta^*),\qquad\check G_n=\int_0^1\frac1n\sum_{i\le n}G\big(X_i,\beta^*+t(\bar\beta_n-\beta^*)\big)\dd t,
\]
and $\|\check G_n-G^*\|\le\|\bar\beta_n-\beta^*\|\,n^{-1}\sum_iL(X_i)+\|n^{-1}\sum_iG(X_i,\beta^*)-G^*\|\to0$ almost surely by the strong law. By \eqref{eq:bahadur} at $t=1$, $\sqrt n(\bar\beta_n-\beta^*)=-P^*\sqrt n\,\bar g_n(\beta^*)+o_P(1)$, and $\sqrt n\,\bar g_n(\beta^*)\Rightarrow N(0,\Omega^*)$ by the central limit theorem. Combining these with Step~1,
\[
\sqrt n\,\hat g_n=\big(I-G^*P^*\big)\sqrt n\,\bar g_n(\beta^*)+(\check G_n-G^*)\sqrt n(\bar\beta_n-\beta^*)+o_P(1)\Rightarrow(I-G^*P^*)\Omega^{*1/2}\mathcal Z,
\]
with $\mathcal Z\sim N(0,I_{\dg})$. Let $Q=\Omega^{*-1/2}G^*\Sigma^*G^{*\top}\Omega^{*-1/2}$. It is symmetric and idempotent, so it is the orthogonal projection onto the column space of $\Omega^{*-1/2}G^*$, which has dimension $\db$. Since $P^*=\Sigma^*G^{*\top}\Omega^{*-1}$,
\[
\Omega^{*-1/2}(I-G^*P^*)\Omega^{*1/2}=I-Q .
\]
By Theorem~\ref{thm:as}, $H_n$ exists for $n$ large and
\[
W_n-W_n\Phi_nH_n\Phi_n\tr W_n\to\Omega^{*-1}-\Omega^{*-1}G^*\Sigma^*G^{*\top}\Omega^{*-1}=\Omega^{*-1/2}(I-Q)\Omega^{*-1/2}
\]
almost surely. By the continuous mapping theorem,
\[
\mathcal J_n\Rightarrow\mathcal Z\tr(I-Q)\Omega^{*1/2}\,\Omega^{*-1/2}(I-Q)\Omega^{*-1/2}\,\Omega^{*1/2}(I-Q)\mathcal Z=\|(I-Q)\mathcal Z\|^2\sim\chi^2_{\dg-\db},
\]
since $I-Q$ is an orthogonal projection of rank $\dg-\db$. Replacing the projected weight by $W_n\to\Omega^{*-1}$ gives the limit $\mathcal Z\tr(I-Q)(I-Q)\mathcal Z$, which is the same. Hence the unprojected statistic $n\hat g_n\tr W_n\hat g_n$ has the same limit.
\end{proof}

\subsection{Proof of Theorem~\ref{thm:dependent}}\label{sec:proofdep}
Throughout this subsection, the assumptions of the first part of Theorem~\ref{thm:dependent} hold. Here $L$ is the dependence lag; the Lipschitz function of Assumption~\ref{ass:smooth} always appears with its argument, as $L(x)$. Let $D=L+1$ and $K=2L+1$, and set $\F_i=\F_0$ for $i\le0$. Two consequences of Assumption~\ref{ass:ldep} are used repeatedly.
\begin{itemize}[leftmargin=*,itemsep=2pt]
\item[(D1)] For $i\ge D$, $X_i$ is independent of $\F_{i-D}=\sigma(X_j:-n_0<j\le i-D)$. For $i\ge K$, the vector $(X_{i-L},\dots,X_i)$ is independent of $\F_{i-K}$, because its indices are at least $(i-K)+L+1$. By stationarity, for an $\F_{i-D}$-measurable $\Theta$-valued $b$ and a measurable $\varphi$ with $\E\sup_{\beta\in\Theta}|\varphi(X,\beta)|<\infty$, we have $\E[\varphi(X_i,b)\mid\F_{i-D}]=\bar\varphi(b)$ with $\bar\varphi(\beta)=\E\varphi(X,\beta)$. The analogous statement holds for functions of $(X_{i-L},\dots,X_i)$ given $\F_{i-K}$.
\item[(D2)] For each $r\in\{1,\dots,D\}$, the variables $X_r,X_{r+D},X_{r+2D},\dots$ are i.i.d.\ with the distribution of $X$. Indeed, $X_{r+(k+1)D}$ is independent of $\sigma(X_r,\dots,X_{r+kD})$ for every $k$. Applying the strong law of large numbers to each residue class, $n^{-1}\sum_{i\le n}\varphi(X_i)\to\E\varphi(X)$ almost surely for every integrable $\varphi$. In particular (T6) applies to $y_j=\varphi(X_j)\ge0$, and (T5) applies to $Y_i=B(X_i)$.
\end{itemize}

\begin{lemma}[Delayed martingale differences]\label{lem:delay}
Let $k\ge1$ be fixed, and let $\xi_i\in\R^d$ be $\F_i$-measurable with $\E[\xi_i\mid\F_{i-k}]=0$ for $i>k$ and $\sup_i\E\|\xi_i\|^p\le c_p$ for some $p\ge2$. For each $r\in\{1,\dots,k\}$, the sequence $(\xi_{r+jk})_{j\ge1}$ is a martingale difference sequence for the filtration $(\F_{r+jk})_{j\ge0}$. Consequently:
\begin{enumerate}[label=(\alph*),itemsep=2pt]
\item $\sum_{i\le n}\xi_i=O(n^{1/2+\varepsilon'})$ almost surely for every $\varepsilon'>0$;
\item for deterministic weights $w_i$, $\E\max_{l\le t}\|\sum_{i\le l}w_i\xi_i\|^p\le C_{p,d,k}\,c_p\big(\sum_{i\le t}w_i^2\big)^{p/2}$;
\item for deterministic weights with $\sum_iw_i^2<\infty$, $\sum_iw_i\xi_i$ converges almost surely;
\item if $\E[\|\xi_i\|^2\mid\F_{i-k}]\le\vartheta_{i-k}$ for $\F_{i-k}$-measurable $\vartheta_{i-k}\ge0$ with $n^{-1}\sum_{i\le n}\vartheta_{i-k}\to0$ in probability, then $n^{-1/2}\max_{l\le n}\|\sum_{i\le l}\xi_i\|\to0$ in probability.
\end{enumerate}
\end{lemma}

\begin{proof}
For $j\ge1$, $\xi_{r+jk}$ is $\F_{r+jk}$-measurable, and $\E[\xi_{r+jk}\mid\F_{r+(j-1)k}]=0$ because $r+(j-1)k=(r+jk)-k$ and $r+jk>k$. Write $\sum_{i\le l}w_i\xi_i=\sum_{i\le k\wedge l}w_i\xi_i+\sum_{r=1}^kM^{(r)}_l$ with $M^{(r)}_l=\sum_{j\ge1:\,r+jk\le l}w_{r+jk}\xi_{r+jk}$. Each $M^{(r)}$ is a martingale in the index $j$.

(a) Take $w_i=1$. By (T2), each $M^{(r)}$, which has at most $n/k$ terms up to $n$, is $O(n^{1/2+\varepsilon'})$, and the first $k$ terms are a fixed finite sum.

(b) Since $(\sum_{r=0}^ka_r)^p\le(k+1)^{p-1}\sum_ra_r^p$, it suffices to bound each part. For the martingales, (T3) in its maximal form and \eqref{eq:burkholder} give $\E\max_{l\le t}\|M^{(r)}_l\|^p\le C_{p,d}c_p\big(\sum_{i\le t}w_i^2\big)^{p/2}$. For the first $k$ terms, H\"older's inequality gives $\E\max_{l\le k}\|\sum_{i\le l}w_i\xi_i\|^p\le k^{p-1}\sum_{i\le k}|w_i|^p\E\|\xi_i\|^p\le k^{p-1}c_p\big(\sum_iw_i^2\big)^{p/2}$.

(c) Each $M^{(r)}$ is a martingale with $\sup_l\E\|M^{(r)}_l\|^2\le c_p^{2/p}\sum_iw_i^2<\infty$, and therefore converges almost surely by the martingale convergence theorem.

(d) The predictable quadratic variation of $M^{(r)}$ (with $w_i=1$) up to time $n$ is at most $\sum_{i\le n}\vartheta_{i-k}$. By (T4) applied to $n^{-1/2}M^{(r)}$, $n^{-1/2}\max_{l\le n}\|M^{(r)}_l\|\to0$ in probability for each $r$. The first $k$ terms are negligible.
\end{proof}

\begin{lemma}[Jacobian and long-run weights]\label{lem:precondLR}
Almost surely:
\begin{enumerate}[label=(\alph*),itemsep=2pt]
\item $\dist(\Phi_i,\mathcal G)\to0$, and the eigenvalues of $S_i$ lie in $[\lambda_0/2,2\Lambda]$ for all $i$ large, so that $W_i=S_i^{-1}\in\mathcal W$;
\item for every $\varepsilon'>0$,
\[
\Big\|\Phi_i-\frac1i\sum_{j=1}^i\bar G(\bar\beta_{j-1})\Big\|+\Big\|S_i-\frac1i\sum_{j=1}^i\Omega_{\rm LR}(\bar\beta_{j-1})\Big\|=O\big(i^{-1/2+\varepsilon'}\big);
\]
\item if $\bar\beta_i\to\beta^*$, then $\Phi_i\to G^*$, $S_i\to\Omega^*_{\rm LR}$, $W_i\to\Omega_{\rm LR}^{*-1}$, $H_i\to\Sigma^*_{\rm LR}=(G^{*\top}\Omega_{\rm LR}^{*-1}G^*)^{-1}$ and $P(\Phi_i,W_i)\to P^*_{\rm LR}=\Sigma^*_{\rm LR}G^{*\top}\Omega_{\rm LR}^{*-1}$;
\item for all $i$ beyond a finite random index, $\|\Phi_i-\Phi_{i-1}\|+\|W_i-W_{i-1}\|\le Ci^{2/\kappa-1}$.
\end{enumerate}
\end{lemma}

\begin{proof}
We first note that $\Omega_{\rm LR}$ is Lipschitz on $\Theta$. Indeed, for each $|h|\le L$,
\[
\big\|\E[g(X_0,\beta)g(X_h,\beta)\tr]-\E[g(X_0,\beta')g(X_h,\beta')\tr]\big\|\le2\E[B(X_0)B(X_h)]\|\beta-\beta'\|\le2\E[B(X)^2]\|\beta-\beta'\|,
\]
so $\Omega_{\rm LR}$ is Lipschitz with constant $2(2L+1)\E[B(X)^2]$. Also $\bar G$ is Lipschitz with constant $\bar L$.

(b), Jacobian. As in Lemma~\ref{lem:precond}, $\Phi_i=(n_0+i)^{-1}\big(n_0\Phi_0+\sum_{j\le i}G(X_j,\bar\beta_{j-1})\big)$. For $j>D$ write
\[
\begin{aligned}G(X_j,\bar\beta_{j-1})-\bar G(\bar\beta_{j-1})&=\underbrace{G(X_j,\bar\beta_{j-D})-\bar G(\bar\beta_{j-D})}_{\Delta^\Phi_j}\\&\quad+\underbrace{G(X_j,\bar\beta_{j-1})-G(X_j,\bar\beta_{j-D})}_{\text{norm}\le L(X_j)\|\bar\beta_{j-1}-\bar\beta_{j-D}\|}\\&\quad-\underbrace{\big(\bar G(\bar\beta_{j-1})-\bar G(\bar\beta_{j-D})\big)}_{\text{norm}\le\bar L\|\bar\beta_{j-1}-\bar\beta_{j-D}\|}.
.\end{aligned}
\]
By (D1), $\E[\Delta^\Phi_j\mid\F_{j-D}]=0$, and $\E\|\Delta_j^\Phi\|^2\le4\E[B(X)^2]$. Lemma~\ref{lem:delay}(a) with delay $k=D$ gives $\sum_{j\le i}\Delta^\Phi_j=O(i^{1/2+\varepsilon'})$. By \eqref{eq:barstep}, $\|\bar\beta_{j-1}-\bar\beta_{j-D}\|\le D_\Theta\sum_{l=j-D+1}^{j-1}l^{-1}\le2D\,D_\Theta/j$ for $j\ge2D$. By (T6) and (D2), $\sum_{j\le i}(L(X_j)+\bar L)/j\le C(\omega)(2+\log i)$. Dividing by $n_0+i$ gives the bound for $\Phi_i$, since $\log i/i=O(i^{-1/2})$ and the terms with $j\le2D$ and the pilot term are $O(1/i)$.

(b), weights. From \eqref{eq:lrv}, $S_i=(n_0+i)^{-1}\big(n_0S_0+\sum_{j\le i}y_j\big)$. For $j>K$ let $b_j=\bar\beta_{j-K}$, which is $\F_{j-K}$-measurable, and let $\tilde y_j$ be $y_j$ with every $s_{j-h}=g(X_{j-h},\bar\beta_{j-h-1})$, $0\le h\le L$, replaced by $\tilde s_{j-h}=g(X_{j-h},b_j)$. By (D1) and stationarity,
\[
\begin{aligned}\E[\tilde y_j\mid\F_{j-K}]&=\Big(\E[g(X_0,\beta)g(X_0,\beta)\tr]\\&\qquad+\sum_{h=1}^L\big(\E[g(X_h,\beta)g(X_0,\beta)\tr]+\E[g(X_0,\beta)g(X_h,\beta)\tr]\big)\Big)\Big|_{\beta=b_j}=\Omega_{\rm LR}(b_j).\end{aligned}
\]
Moreover $\|\tilde y_j\|\le(2L+1)\big(\sum_{h=0}^LB(X_{j-h})\big)^2$, so $\E\|\tilde y_j-\Omega_{\rm LR}(b_j)\|^2\le C\E[B(X)^4]$. Lemma~\ref{lem:delay}(a) with delay $k=K$ gives $\sum_{j\le i}\big(\tilde y_j-\Omega_{\rm LR}(b_j)\big)=O(i^{1/2+\varepsilon'})$.

Next, the indices $j-h-1$ for $0\le h\le L$ lie between $j-K$ and $j-1$. So by \eqref{eq:barstep}, $\|\bar\beta_{j-h-1}-b_j\|\le2K\,D_\Theta/j$ for $j\ge2K$. Hence $\|s_{j-h}-\tilde s_{j-h}\|\le2K\,D_\Theta B(X_{j-h})/j$. Since $y_j$ is a sum of $2L+1$ products $ab\tr$ with $\|ab\tr-\tilde a\tilde b\tr\|\le\|a-\tilde a\|\,\|b\|+\|\tilde a\|\,\|b-\tilde b\|$,
\[
\|y_j-\tilde y_j\|\le\frac Cj\sum_{h=0}^LB(X_{j-h})^2\le C\sum_{h=0}^L\frac{B(X_{j-h})^2}{j-h},
\]
and $\sum_{j\le i}\|y_j-\tilde y_j\|\le C(\omega)(2+\log i)$ by (T6) and (D2) applied to $B(X)^2$. Finally, $\|\Omega_{\rm LR}(b_j)-\Omega_{\rm LR}(\bar\beta_{j-1})\|\le C\|\bar\beta_{j-K}-\bar\beta_{j-1}\|\le C/j$. Collecting terms, together with the finitely many terms with $j\le2K$, which involve the pilot values $s_j$, $j\le0$, and contribute $O(1/i)$, proves the bound for $S_i$.

(a) and (c) follow from (b) exactly as in Lemma~\ref{lem:precond}. The eigenvalues of $\Omega_{\rm LR}(\beta)$ lie in $[\lambda_0,\Lambda]$ by definition, and $\Omega_{\rm LR}$ and $\bar G$ are continuous. Once $S_i$ is positive definite, $W_i=S_i^{-1}$ by the definition of the weights.

(d) By (a), there is a finite random index after which $\Phi_{i-1}\in\mathcal N$, $\|S_{i-1}\|\le2\Lambda$ and $W_{i-1},W_i\in\mathcal W$. By (T5) and (D2), also $B(X_l)\le l^{1/\kappa}$ for $l\ge i-L$. Then
\begin{gather*}
\|\Phi_i-\Phi_{i-1}\|=\frac{\|G(X_i,\bar\beta_{i-1})-\Phi_{i-1}\|}{n_0+i}\le\frac{i^{1/\kappa}+\phi_{\max}}{i},\\ \|S_i-S_{i-1}\|=\frac{\|y_i-S_{i-1}\|}{n_0+i}\le\frac{(2L+1)^3i^{2/\kappa}+2\Lambda}{i},
\end{gather*}
and $\|W_i-W_{i-1}\|=\|W_i(S_{i-1}-S_i)W_{i-1}\|\le4\lambda_0^{-2}\|S_i-S_{i-1}\|$.
\end{proof}

\begin{proof}[Proof of Theorem~\ref{thm:dependent}]
\emph{Step 1: eventual regime and decomposition.} By Lemma~\ref{lem:precondLR}(a), \eqref{eq:Pbound}, $\bar p_i\to\infty$ and (T5), there is a finite random index $\iota_0$ such that for all $i\ge\iota_0$: $\Phi_i\in\mathcal N$, $W_i\in\mathcal W$, the truncation is inactive, $P_i=P(\Phi_i,W_i)$, $\|P_i\|\le p_{\max}$, $B(X_i)\le i^{1/\kappa}$, and the bound of Lemma~\ref{lem:precondLR}(d) holds.

For $i\ge\iota_0+D$, all indices $i-D,\dots,i-1$ belong to this regime. Write $g_i=g(X_i,\beta_{i-1})$ and
\[
e_i=g(X_i,\beta_{i-D})-m(\beta_{i-D}),\qquad r_i=P_{i-D}\big(g(X_i,\beta_{i-1})-g(X_i,\beta_{i-D})\big)+(P_{i-1}-P_{i-D})g_i ,
\]
so that
\begin{equation}\label{eq:delaydec}
P_{i-1}g_i=P_{i-D}m(\beta_{i-D})+P_{i-D}e_i+r_i .
\end{equation}
By (D1), $\E[e_i\mid\F_{i-D}]=0$ and $\E[\|e_i\|^\kappa\mid\F_{i-D}]\le2^\kappa\E[B(X)^\kappa]$. Since $\Pi_\Theta$ is non-expansive and $\beta_{l-1}\in\Theta$, $\|\beta_l-\beta_{l-1}\|\le\gamma_l\|P_{l-1}g(X_l,\beta_{l-1})\|\le p_{\max}\gamma_ll^{1/\kappa}$ in the regime. Since $\gamma_{i-L}/\gamma_i$ is bounded,
\begin{equation}\label{eq:stepdiff}
\|\beta_{i-1}-\beta_{i-D}\|\le\sum_{l=i-D+1}^{i-1}p_{\max}\gamma_ll^{1/\kappa}\le C\gamma_ii^{1/\kappa}.
\end{equation}
By \eqref{eq:Plip} and Lemma~\ref{lem:precondLR}(d), $\|P_{i-1}-P_{i-D}\|\le C_P\sum_{l=i-D+1}^{i-1}\big(\|\Phi_l-\Phi_{l-1}\|+\|W_l-W_{l-1}\|\big)\le Ci^{2/\kappa-1}$. With $\|g(X_i,\beta)-g(X_i,\beta')\|\le B(X_i)\|\beta-\beta'\|$ and $\|g_i\|\le B(X_i)\le i^{1/\kappa}$,
\begin{equation}\label{eq:rbound}
\|r_i\|\le C\big(\gamma_ii^{2/\kappa}+i^{3/\kappa-1}\big).
\end{equation}

\emph{Step 2: almost sure convergence.} Let $V_i=\|\delta_i\|^2$. For $i\ge\iota_0+D$, non-expansiveness of $\Pi_\Theta$ gives $V_i\le V_{i-1}-2\gamma_i\langle\delta_{i-1},P_{i-1}g_i\rangle+\gamma_i^2\|P_{i-1}g_i\|^2$. By \eqref{eq:delaydec},
\[
\langle\delta_{i-1},P_{i-1}g_i\rangle=\langle\delta_{i-D},P_{i-D}m(\beta_{i-D})\rangle+\chi_i+\langle\delta_{i-D},r_i\rangle+\langle\delta_{i-1}-\delta_{i-D},P_{i-1}g_i\rangle,\qquad\chi_i=\langle\delta_{i-D},P_{i-D}e_i\rangle .
\]
The first term is at least $cV_{i-D}$ by Assumption~\ref{ass:mono}. By \eqref{eq:stepdiff} and \eqref{eq:rbound}, the third and fourth terms are bounded in absolute value by $C(\gamma_ii^{2/\kappa}+i^{3/\kappa-1})$, and $\gamma_i^2\|P_{i-1}g_i\|^2\le p_{\max}^2\gamma_i^2i^{2/\kappa}$. Hence, for $i\ge\iota_0+D$,
\begin{equation}\label{eq:pathRS}
V_i\le V_{i-1}-2c\gamma_iV_{i-D}-2\gamma_i\chi_i+b_i,\qquad b_i=C\big(\gamma_i^2i^{2/\kappa}+\gamma_ii^{3/\kappa-1}\big).
\end{equation}
The sequence $b_i$ is summable. First, $\gamma_i^2i^{2/\kappa}\le Ci^{-2a+2/\kappa}$ and $2a-2/\kappa>1$ because $\kappa>2/(2a-1)$. Second, $\gamma_ii^{3/\kappa-1}\le Ci^{-1-(a-3/\kappa)}$, and $a>3/\kappa$. To see the last inequality: if $a\le3/4$, then $2/(2a-1)\ge3/a$ (equivalent to $2a\ge6a-3$), so $\kappa>3/a$; if $a>3/4$, then $3/a<4\le\kappa$.

Let $\tilde\chi_i=\langle\delta_{i-D},\tilde P_{i-D}e_i\rangle$, with $\tilde P_{i-D}=P_{i-D}\mathbf 1\{\|P_{i-D}\|\le p_{\max}\}$. Then $\tilde\chi_i$ is $\F_i$-measurable, $\E[\tilde\chi_i\mid\F_{i-D}]=0$ and $\E\tilde\chi_i^2\le4D_\Theta^2p_{\max}^2\E[B(X)^2]$. Set $\tilde\chi_i=0$ for $i\le D$. By Lemma~\ref{lem:delay}(c) with delay $k=D$ and $w_i=\gamma_i$, $\sum_i\gamma_i\tilde\chi_i$ converges almost surely. Since $\tilde\chi_i=\chi_i$ for $i\ge\iota_0+D$, the series $\sum_{i\ge\iota_0+D}\gamma_i\chi_i$ converges.

Let $\tau_i=\sum_{j>i}\gamma_j\chi_j$ and $\bar b_i=\sum_{j>i}b_j$, both tending to zero, and let $U_i=V_i-2\tau_i+\bar b_i$. Since $\gamma_i\chi_i=\tau_{i-1}-\tau_i$ and $b_i=\bar b_{i-1}-\bar b_i$, \eqref{eq:pathRS} becomes
\[
U_i\le U_{i-1}-2c\gamma_iV_{i-D}\qquad(i\ge\iota_0+D).
\]
Thus $U_i$ is non-increasing and bounded below by $-2\sup_j|\tau_j|>-\infty$. Hence $U_i$ converges, and $\sum_i\gamma_iV_{i-D}<\infty$. Consequently $V_i=U_i+2\tau_i-\bar b_i$ converges. Since $\sum_i\gamma_i=\infty$, $\liminf_iV_{i-D}=0$, so $V_i\to0$. Thus $\beta_i\to\beta^*$ and $\bar\beta_i\to\beta^*$. Lemma~\ref{lem:precondLR}(c) gives the limits of $\Phi_n$, $W_n$, $H_n$ and $P_n$. The projection is eventually inactive by the argument of Step~3 of the proof of Theorem~\ref{thm:as}, which uses only $B(X_i)\le i^{1/\kappa}$ and $a>1/\kappa$.

\emph{Step 3: rates.} Fix $\varepsilon>\max\{(1-a)/\kappa,\,2/\kappa-a/2,\,3/\kappa-1+a/2\}$. Without loss of generality $\varepsilon<a/2$. This is possible because each of the three lower bounds is smaller than $a/2$: $(1-a)/\kappa<a/2$ since $\kappa\ge4>2(1-a)/a$; $2/\kappa-a/2<a/2$ since $\kappa\ge4>2/a$; and $3/\kappa-1+a/2<a/2$ since $\kappa>3$.

Let $v_i=i^{-a/2+\varepsilon}$, and let $\iota_1\ge\iota_0+D$ be a finite random index after which the projection is inactive and $\|A_{i-D}-I\|\le1/2$, where $A_{i-D}=P_{i-D}\tilde G(\beta_{i-D})\to I$. For $i>\iota_1$, \eqref{eq:delaydec} and $P_{i-D}m(\beta_{i-D})=A_{i-D}\delta_{i-D}$ give
\[
\delta_i=(I-\gamma_iA_{i-D})\delta_{i-1}-\gamma_iP_{i-D}e_i+\gamma_if_i,\qquad f_i=A_{i-D}(\delta_{i-1}-\delta_{i-D})-r_i .
\]
By \eqref{eq:stepdiff} and \eqref{eq:rbound}, $\|f_i\|\le C(\gamma_ii^{2/\kappa}+i^{3/\kappa-1})$.

Enlarging $i_1$ if necessary (Lemma~\ref{lem:weights} holds for every larger $i_1$), we assume $i_1\ge D$. Define $Z_i$ as in Step~1 of the proof of Theorem~\ref{thm:rate}, with $\xi_i=\tilde P_{i-D}e_i$ for $i>D$ and $\xi_i=0$ for $i\le D$. These satisfy $\E[\xi_i\mid\F_{i-D}]=0$ and $\E\|\xi_i\|^\kappa\le p_{\max}^\kappa2^\kappa\E[B(X)^\kappa]$. Lemma~\ref{lem:delay}(b) with delay $k=D$ replaces Doob's and Burkholder's inequalities in that step, so the bound $\E\max_{l\le t}\|M_l\|^\kappa\le C(\sum_{j\le t}\gamma_j^2/\Gamma_{i_1,j}^2)^{\kappa/2}$ still holds, and the rest of Step~1 is unchanged. Hence $\|Z_i\|=O(v_i)$. For $i>\iota_1$, $\Delta_i=\delta_i-Z_i$ satisfies
\[
\Delta_i=(I-\gamma_iA_{i-D})\Delta_{i-1}-\gamma_i(A_{i-D}-I)Z_{i-1}+\gamma_if_i ,
\]
hence $\|\Delta_i\|\le(1-\gamma_i/2)\|\Delta_{i-1}\|+\gamma_i\eta_iv_{i-1}$ with
\[
\eta_i=\frac{\|A_{i-D}-I\|\,\|Z_{i-1}\|+\|f_i\|}{v_{i-1}}\le o(1)+C\big(i^{-a/2+2/\kappa-\varepsilon}+i^{3/\kappa-1+a/2-\varepsilon}\big)\to0 ,
\]
by the choice of $\varepsilon$. Lemma~\ref{lem:det} gives $\|\Delta_i\|=o(v_i)$, hence $\|\delta_i\|=O(v_i)$, and the rates for $\bar\beta_n$, $\Phi_n$ and $W_n$ follow as in Steps~2 and~3 of the proof of Theorem~\ref{thm:rate}, using Lemma~\ref{lem:precondLR}(b) and the Lipschitz continuity of $\bar G$ and $\Omega_{\rm LR}$.

\emph{Step 4: functional central limit theorem.} Assume $\kappa>\kappa_L(a)$. We claim that $\varepsilon$ can be chosen in Step~3 with $\varepsilon<\min\{(1-a)/2,(2a-1)/4\}$. It suffices that each lower bound of Step~3 is below both quantities:
\begin{itemize}[leftmargin=*,itemsep=1pt]
\item $(1-a)/\kappa<(1-a)/2$ and $(1-a)/\kappa<(2a-1)/4$, by $\kappa>\kappa(a)$;
\item $2/\kappa-a/2<(1-a)/2$, which is equivalent to $\kappa>4$; and $2/\kappa-a/2<(2a-1)/4$, which is equivalent to $\kappa>8/(4a-1)$ and follows from $\kappa>2/(a-1/2)=4/(2a-1)>8/(4a-1)$;
\item $3/\kappa-1+a/2<(1-a)/2$, which is equivalent to $\kappa>6/(3-2a)$ and follows from $\kappa>6$; and $3/\kappa-1+a/2<(2a-1)/4$, which is equivalent to $\kappa>4$.
\end{itemize}
With such an $\varepsilon$, the bounds used for $R^{(1)}$ and $R^{(4)}$ in the proof of Theorem~\ref{thm:fclt} hold. For $i>\iota_1$, rearranging the recursion of Step~3 gives
\[
\delta_{i-D}=\frac{\delta_{i-1}-\delta_i}{\gamma_i}-P_{i-D}e_i-\rho_{i-D}-r_i,\qquad\rho_{i-D}=P_{i-D}m(\beta_{i-D})-\delta_{i-D}.
\]
Write $e_i=\xi^0_i+\zeta_i$ with $\xi^0_i=g(X_i,\beta^*)$ and $\zeta_i=g(X_i,\beta_{i-D})-g(X_i,\beta^*)-m(\beta_{i-D})$. By (D1), $\E[\zeta_i\mid\F_{i-D}]=0$ and $\E[\|\zeta_i\|^2\mid\F_{i-D}]\le\E[B(X)^2]\|\delta_{i-D}\|^2$. Summing over $\iota_1<i\le k$,
\[
\sum_{i=\iota_1+1}^k\delta_{i-D}=-P^*_{\rm LR}\sum_{i=\iota_1+1}^k\xi^0_i+R^{(1)}_k-R^{(2)}_k-R^{(3)}_k-R^{(4)}_k-R^{(5)}_k,
\]
with $R^{(1)}$ as before, $R^{(2)}_k=P^*_{\rm LR}\sum\zeta_i$, $R^{(3)}_k=\sum(P_{i-D}-P^*_{\rm LR})e_i$, $R^{(4)}_k=\sum\rho_{i-D}$ and $R^{(5)}_k=\sum r_i$, all sums over $\iota_1<i\le k$.

The left-hand side equals $\sum_{j=\iota_1+1-D}^{k-D}\delta_j$, whereas $\sum_{i=\iota_1+1}^k\delta_{i-1}=\sum_{j=\iota_1}^{k-1}\delta_j$. The two differ by at most $2(D-1)$ terms, each bounded by $D_\Theta$, which is negligible after division by $\sqrt n$. The terms $R^{(1)}$ and $R^{(4)}$ are bounded as in the proof of Theorem~\ref{thm:fclt}, using $\|\rho_{i-D}\|\le Cv_{i-D}^2\le Cv_i^2$. For $R^{(5)}$, \eqref{eq:rbound} gives
\[
\frac1{\sqrt n}\sum_{i\le n}\|r_i\|\le\frac{C(\omega)}{\sqrt n}\Big(1+n^{1-a+2/\kappa}+n^{3/\kappa}\Big)\to0 ,
\]
because $1/2-a+2/\kappa<0$ (as $\kappa>2/(a-1/2)$) and $3/\kappa<1/2$ (as $\kappa>6$). For $R^{(2)}$, apply Lemma~\ref{lem:delay}(d) with delay $k=D$ to the sequence $P^*_{\rm LR}\zeta_i$, $i>D$, with $\vartheta_{i-D}=C\|\delta_{i-D}\|^2$, whose Ces\`aro averages tend to zero. As in the proof of Theorem~\ref{thm:fclt}, $\max_{k\le n}\|R^{(2)}_k\|$ is at most twice the maximum of the full partial sums. For $R^{(3)}$, apply Lemma~\ref{lem:delay}(d) with delay $k=D$ to $(\tilde P_{i-D}-P^*_{\rm LR})e_i$, with $\vartheta_{i-D}=4\E[B(X)^2]\|\tilde P_{i-D}-P^*_{\rm LR}\|^2\to0$, and note that $\tilde P_{i-D}=P_{i-D}$ for $i>\iota_1$. Hence \eqref{eq:bahadur} holds with $P^*_{\rm LR}$ in place of $P^*$.

It remains to show that $n^{-1/2}\sum_{i\le\lfloor nt\rfloor}\xi^0_i\Rightarrow\Omega_{\rm LR}^{*1/2}\mathbf B_g$ in $D[0,1]^{\dg}$. Since $P^*_{\rm LR}\Omega^*_{\rm LR}P^{*\top}_{\rm LR}=\Sigma^*_{\rm LR}$, the conclusion of Theorem~\ref{thm:fclt} then follows as before. The statement concerns only the distribution of the stationary, $L$-dependent, centred and square-integrable sequence $\xi^0_i$. We may therefore assume, by Kolmogorov's extension theorem, that $(X_i)$ is defined for all $i\in\mathbb Z$, stationary and $L$-dependent. Let $\mathcal H_i=\sigma(X_j:j\le i)$. By $L$-dependence, $\E[\xi^0_{i+h}\mid\mathcal H_i]=0$ for $h>L$. Define
\[
\upsilon_i=\sum_{h=1}^L\E[\xi^0_{i+h}\mid\mathcal H_i],\qquad d_i=\sum_{h=0}^L\big(\E[\xi^0_{i+h}\mid\mathcal H_i]-\E[\xi^0_{i+h}\mid\mathcal H_{i-1}]\big).
\]
Then $(d_i)$ is a stationary, square-integrable martingale difference sequence for $(\mathcal H_i)$, and $(\upsilon_i)$ is stationary and square-integrable. Using $\E[\xi^0_{i+L}\mid\mathcal H_{i-1}]=0$,
\[
d_i=\xi^0_i+\upsilon_i-\sum_{h=1}^{L+1}\E[\xi^0_{i-1+h}\mid\mathcal H_{i-1}]=\xi^0_i+\upsilon_i-\upsilon_{i-1},\quad\text{i.e.}\quad\xi^0_i=d_i+\upsilon_{i-1}-\upsilon_i .
\]
The coboundary is negligible. For $\eta>0$, stationarity gives
\[
\Prob\Big(\max_{k\le n}\|\upsilon_k\|>\eta\sqrt n\Big)\le n\Prob(\|\upsilon_0\|>\eta\sqrt n)\le\eta^{-2}\E\big[\|\upsilon_0\|^2\mathbf 1\{\|\upsilon_0\|>\eta\sqrt n\}\big]\to0 ,
\]
so $n^{-1/2}\max_{k\le n}\|\upsilon_0-\upsilon_k\|\to0$ in probability.

The shift is mixing for an $L$-dependent stationary sequence, because events depending on finitely many coordinates become independent of their shifts once the shift exceeds $L$ plus the length of their index window. It is therefore ergodic. The sequence $\E[d_id_i\tr\mid\mathcal H_{i-1}]$ is stationary, so by the ergodic theorem $n^{-1}\sum_{i\le n}\E[d_id_i\tr\mid\mathcal H_{i-1}]\to\E[d_0d_0\tr]$ almost surely. The conditional Lindeberg condition holds because $\E\big[n^{-1}\sum_{i\le n}\E[\|d_i\|^2\mathbf 1\{\|d_i\|>\eta\sqrt n\}\mid\mathcal H_{i-1}]\big]=\E\big[\|d_0\|^2\mathbf 1\{\|d_0\|>\eta\sqrt n\}\big]\to0$. By the martingale invariance principle \cite{Brown1971}, applied to $\lambda\tr d_i$ for every $\lambda$, together with the Cram\'er--Wold device and coordinatewise tightness, $n^{-1/2}\sum_{i\le\lfloor nt\rfloor}d_i\Rightarrow(\E d_0d_0\tr)^{1/2}\mathbf B_g$.

Finally, $\E[d_0d_0\tr]=\Omega^*_{\rm LR}$. Indeed, $\operatorname{Var}(\sum_{i\le n}\xi^0_i)=\sum_{|h|\le L}(n-|h|)\E[\xi^0_0\xi^{0\top}_h]=n\Omega^*_{\rm LR}+O(1)$. On the other hand, $\sum_{i\le n}\xi^0_i=\sum_{i\le n}d_i+\upsilon_0-\upsilon_n$, whose variance is $n\E[d_0d_0\tr]+O(\sqrt n)$ by the Cauchy--Schwarz inequality.

\emph{Step 5: inference and the overidentification statistic.} Corollary~\ref{cor:inference} follows from Step~4 and $H_n\to\Sigma^*_{\rm LR}$, with the same proof. In the proof of Theorem~\ref{thm:J}, Step~1 uses only the rates of Step~3 with $\varepsilon<(2a-1)/4$ and the strong law for $L(X_i)$, which holds by (D2). In Step~2, $\check G_n\to G^*$ by (D2), and $\sqrt n\,\bar g_n(\beta^*)\Rightarrow N(0,\Omega^*_{\rm LR})$ by Step~4 at $t=1$. The rest of the argument is unchanged with $\Omega^*_{\rm LR}$, $\Sigma^*_{\rm LR}$ and $P^*_{\rm LR}$ in place of $\Omega^*$, $\Sigma^*$ and $P^*$, because $W_n\to\Omega_{\rm LR}^{*-1}$.
\end{proof}

%=====================================================================
\section{Numerical experiments}\label{sec:numerics}
%=====================================================================
All computations were done in Python with NumPy and SciPy, and in C for the compiled implementation. The codes used in this paper can be found here : \text{https://github.com/nducduy/StochasticGMM}. We compare SGMM (Algorithm~\ref{alg:sgmm}) with two other online methods and two offline estimators. The online methods are online 2SLS, i.e.\ Algorithm~\ref{alg:sgmm} with $s_i=z_i$ (Remark~\ref{rem:design}(iv)), and a first-order version that uses the step $\gamma_i\Phi_{i-1}\tr W_{i-1}g(X_i,\beta_{i-1})$ without the factor $H_{i-1}$ but with the same efficient weights and averaging. The offline estimators are 2SLS-weighted GMM and two-step efficient GMM, both computed on the full sample by damped Gauss--Newton iterations to convergence. All online methods start from the same pilot estimate, a 2SLS-weighted GMM estimate on $n_0=1000$ observations, and use the default steps $\gamma_i=2n_0^{-1/4}(n_0+i)^{-3/4}$ of Remark~\ref{rem:design}(v), the projection onto $\Theta=[-10,10]^3$, the truncation $\bar p_i=10^8i^{0.1}$ and the ridge $\eta_0=10^{-6}$; the offline estimators use the same $n_0+n$ observations. Results are recorded at $n=10^4,10^5,10^6$ along a single pass over $10^6$ observations, for $1000$ independent replications. The online methods were run with a compiled implementation of Algorithm~\ref{alg:sgmm}, which we checked against a separate Python implementation (identical results up to rounding errors of order $10^{-15}$). Neither the truncation nor the projection was ever active in Sections~\ref{sec:numlin}--\ref{sec:numJ}.

\subsection{Linear IV regression}\label{sec:numlin}
In the model of Example~\ref{ex:linear} we take $x=(1,x_1,x_2)\tr$, $z=(1,z_1,\dots,z_7)\tr$ with $z_k$ i.i.d.\ $N(0,1)$, $x_1=0.3(z_1+z_2+z_3)+v_1$, $x_2=0.3(z_4+z_5+z_6)+v_2$, and $u=\big(0.5(v_1+v_2)/\sqrt2+\sqrt{0.75}\,e\big)\exp\big((z_1+z_4)/2-1/2\big)$, with $v_1,v_2,e$ i.i.d.\ $N(0,1)$ and $\beta^*=(1,0.5,-0.5)\tr$. Both regressors are endogenous, and the errors are conditionally heteroskedastic, so that efficient GMM improves on 2SLS, mostly for the intercept. Table~\ref{tab:lin} and Figure~\ref{fig:rmse}(a) show the results.

\begin{table}[t]\centering
\caption{Linear IV model (Example~\ref{ex:linear}), $\db=3$, $\dg=8$, $1000$ replications. RMSE of the estimators, empirical coverage of nominal $95\%$ plug-in confidence intervals (SGMM: $H_n/n$; online and offline 2SLS: sandwich), and coverage of random-scaling (RS) intervals for the coefficient $\beta_1$ of the endogenous regressor $x_1$. Online methods use a pilot sample of $n_0=1000$ and the default steps of Remark~\ref{rem:design}(v); offline methods use the same $n_0+n$ observations. The Monte Carlo standard error of a coverage of $0.95$ is $0.007$.}\label{tab:lin}
\footnotesize\setlength{\tabcolsep}{3.5pt}
\begin{tabular}{@{}cl ccc ccc c@{}}\toprule
& & \multicolumn{3}{c}{RMSE $\times10^3$} & \multicolumn{3}{c}{coverage, plug-in} & RS\\ \cmidrule(lr){3-5}\cmidrule(lr){6-8}\cmidrule(l){9-9}
$n$ & method & $\beta_0$ & $\beta_1$ & $\beta_2$ & $\beta_0$ & $\beta_1$ & $\beta_2$ & $\beta_1$\\ \midrule
$10^4$ & SGMM (efficient) & 7.11 & 22.03 & 22.39 & .937 & .927 & .922 & .824 \\
 & online 2SLS & 10.26 & 23.52 & 23.81 & .944 & .943 & .935 & .814 \\
 & first-order & 8.66 & 34.94 & 35.74 & -- & -- & -- & .482 \\
 & offline 2SLS & 9.55 & 21.63 & 22.22 & .948 & .943 & .939 & -- \\
 & offline two-step GMM & 6.10 & 19.62 & 20.17 & .955 & .946 & .943 & -- \\
\midrule
$10^5$ & SGMM (efficient) & 2.18 & 6.98 & 6.82 & .942 & .932 & .940 & .882 \\
 & online 2SLS & 3.21 & 7.37 & 7.40 & .952 & .933 & .934 & .881 \\
 & first-order & 2.68 & 10.77 & 10.97 & -- & -- & -- & .757 \\
 & offline 2SLS & 3.07 & 7.08 & 7.05 & .954 & .940 & .950 & -- \\
 & offline two-step GMM & 2.06 & 6.63 & 6.45 & .953 & .937 & .957 & -- \\
\midrule
$10^6$ & SGMM (efficient) & 0.70 & 2.20 & 2.13 & .941 & .939 & .945 & .905 \\
 & online 2SLS & 1.02 & 2.36 & 2.26 & .950 & .930 & .948 & .908 \\
 & first-order & 0.75 & 2.65 & 2.49 & -- & -- & -- & .847 \\
 & offline 2SLS & 1.00 & 2.30 & 2.21 & .951 & .935 & .951 & -- \\
 & offline two-step GMM & 0.68 & 2.14 & 2.08 & .945 & .947 & .956 & -- \\
\bottomrule\end{tabular}\end{table}

SGMM approaches the accuracy of offline two-step GMM as $n$ grows. The ratios of RMSEs are between $1.11$ and $1.17$ at $n=10^4$, between $1.05$ and $1.06$ at $n=10^5$, and at most $1.03$ at $n=10^6$; online 2SLS behaves in the same way relative to offline 2SLS. The efficiency gain of GMM over 2SLS, a factor of $1.5$ in the RMSE of the intercept, is achieved online. The plug-in intervals based on $H_n/n$ have coverage between $0.922$ and $0.945$, the lower values occurring at $n=10^4$. The random-scaling intervals undercover at these sample sizes ($0.82$ at $n=10^4$, $0.88$ at $10^5$ and $0.91$ at $10^6$). They are justified by the functional central limit theorem, which describes the partial sums of the iterates over a long horizon; with the small default steps, consecutive iterates are strongly correlated, and the effective number of independent blocks in the path is small even at $n=10^6$. We therefore recommend the plug-in intervals, which come for free with SGMM. The first-order version is $1.2$ to $1.6$ times less accurate than SGMM at $n=10^4$ and $10^5$, and $1.07$ to $1.2$ times at $n=10^6$.

\subsection{Exponential IV regression}\label{sec:numexp}
For the multiplicative model of Example~\ref{ex:nonlinear} we take $x=(1,x_1,x_2)\tr$ and $z=(1,z_1,\dots,z_4,x_2)\tr$, with $z_k,x_2$ i.i.d.\ $N(0,1)$, $x_1=0.25(z_1+\dots+z_4)+v$, $v\sim N(0,0.25)$, and
\[
\eta=\tfrac12\exp\big(w/2-1/8\big)+\tfrac12\,\Gamma,\qquad w=1.2\,v+0.8\,e,
\]
where $e\sim N(0,1)$ and, conditionally on $z$, $\Gamma$ has a gamma distribution with mean $1$ and standard deviation $s(z)=0.1+\max(z_1+z_2,0)$. Then $\E[\eta\mid z]=1$, $\operatorname{corr}(v,w)=0.6$ makes $x_1$ endogenous, and $s(z)$ makes the errors heteroskedastic. We set $\beta^*=(0.5,0.5,-0.5)\tr$ and $y=\exp(x\tr\beta^*)\eta$.

\begin{table}[t]\centering
\caption{Exponential IV model (Example~\ref{ex:nonlinear}), $\db=3$, $\dg=6$, $1000$ replications; columns as in Table~\ref{tab:lin}.}\label{tab:exp}
\footnotesize\setlength{\tabcolsep}{3.5pt}
\begin{tabular}{@{}cl ccc ccc c@{}}\toprule
& & \multicolumn{3}{c}{RMSE $\times10^3$} & \multicolumn{3}{c}{coverage, plug-in} & RS\\ \cmidrule(lr){3-5}\cmidrule(lr){6-8}\cmidrule(l){9-9}
$n$ & method & $\beta_0$ & $\beta_1$ & $\beta_2$ & $\beta_0$ & $\beta_1$ & $\beta_2$ & $\beta_1$\\ \midrule
$10^4$ & SGMM (efficient) & 4.25 & 15.25 & 6.33 & .926 & .916 & .924 & .824 \\
 & online 2SLS & 5.47 & 16.07 & 6.16 & .939 & .936 & .948 & .822 \\
 & first-order & 4.54 & 19.65 & 6.91 & -- & -- & -- & .720 \\
 & offline 2SLS & 5.03 & 14.52 & 5.69 & .941 & .945 & .952 & -- \\
 & offline two-step GMM & 3.82 & 13.52 & 5.65 & .938 & .929 & .944 & -- \\
\midrule
$10^5$ & SGMM (efficient) & 1.26 & 4.72 & 2.03 & .945 & .945 & .927 & .875 \\
 & online 2SLS & 1.71 & 5.01 & 2.00 & .936 & .958 & .929 & .892 \\
 & first-order & 1.28 & 5.59 & 2.04 & -- & -- & -- & .838 \\
 & offline 2SLS & 1.61 & 4.79 & 1.92 & .943 & .959 & .941 & -- \\
 & offline two-step GMM & 1.19 & 4.42 & 1.92 & .954 & .948 & .944 & -- \\
\midrule
$10^6$ & SGMM (efficient) & 0.40 & 1.45 & 0.61 & .951 & .947 & .940 & .912 \\
 & online 2SLS & 0.55 & 1.61 & 0.61 & .939 & .954 & .940 & .915 \\
 & first-order & 0.40 & 1.55 & 0.61 & -- & -- & -- & .879 \\
 & offline 2SLS & 0.53 & 1.56 & 0.60 & .943 & .955 & .946 & -- \\
 & offline two-step GMM & 0.40 & 1.41 & 0.60 & .956 & .955 & .945 & -- \\
\bottomrule\end{tabular}\end{table}

\begin{figure}[t]
\centering
\includegraphics[width=\textwidth]{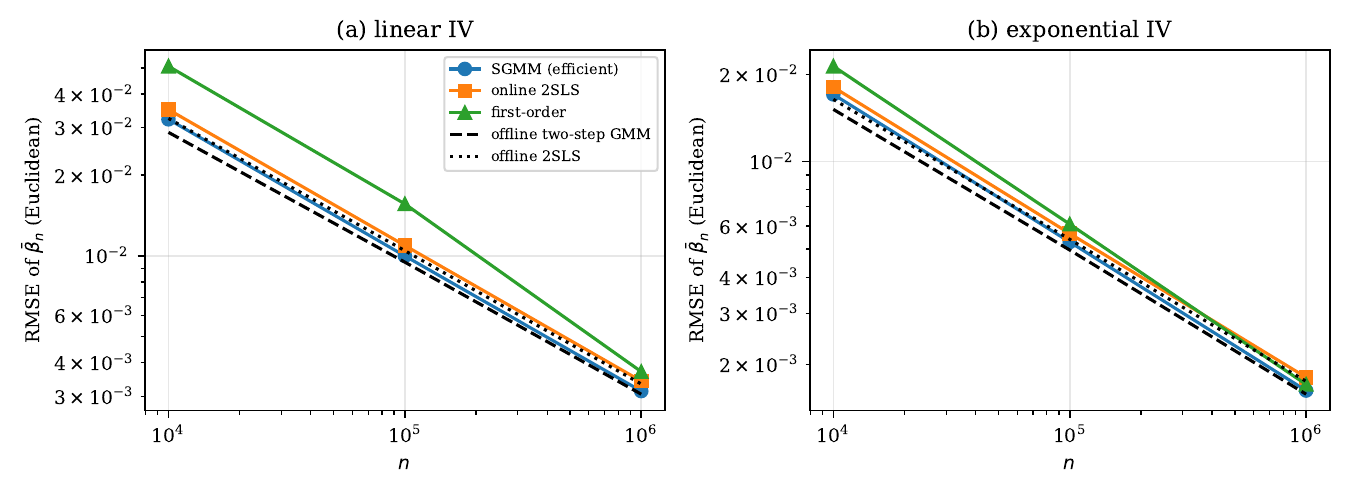}
\caption{Euclidean RMSE of the estimators of Tables~\ref{tab:lin} and~\ref{tab:exp} as a function of the sample size (logarithmic scales): (a) linear model; (b) exponential model.}\label{fig:rmse}
\end{figure}

The conclusions of Table~\ref{tab:exp} and Figure~\ref{fig:rmse}(b) are the same as in the linear model. The RMSEs of SGMM are within $11$--$13\%$ of those of offline two-step GMM at $n=10^4$, within $7\%$ at $n=10^5$ and within $3.1\%$ at $n=10^6$, and the plug-in coverages lie between $0.916$ and $0.951$. The offline estimator required the full sample in memory and a pass over it at every Gauss--Newton iteration, whereas SGMM made one pass with $O(\dg^2+\dg\db)$ memory.

\subsection{Overidentification test}\label{sec:numJ}
Table~\ref{tab:J} reports the rejection frequencies of the online statistic $\mathcal J_n$ of \eqref{eq:Jn} and of the offline Hansen--Sargan statistic. With $1000$ replications, the Monte Carlo standard error of a rejection rate of $5\%$ is $0.7\%$. Under correct specification, the online statistic has rejection rates between $4.0\%$ and $5.6\%$ at the $5\%$ level in both models and at all sample sizes, similar to the offline statistic, in agreement with Theorem~\ref{thm:J}. The unprojected statistic $n\hat g_n\tr W_n\hat g_n$ is oversized in the exponential model at $n=10^4$ ($8.9\%$), which illustrates the benefit of the projection. For the power, we make the instrument $z_4$ invalid by multiplying $\eta$ by $\exp(\vartheta z_4-\vartheta^2/2)$, so that $\E[z_4(\eta-1)]=\vartheta$. With $\vartheta=0.01$, the online and offline tests reject at the $5\%$ level in $22.5\%$ and $23.8\%$ of the replications at $n=10^4$, in $95.5\%$ of the replications at $n=10^5$, and always at $n=10^6$; the online test has essentially the same power as the offline test.

\begin{table}[t]\centering
\caption{Rejection rates of the online statistic $\mathcal J_n$ of \eqref{eq:Jn} and of the offline Hansen--Sargan statistic at the nominal $10\%$, $5\%$ and $1\%$ levels, $1000$ replications. Size: correctly specified linear and exponential models. Power: exponential model in which the instrument $z_4$ is invalid ($\vartheta=0.01$). Last column: unprojected online statistic $n\hat g_n\tr W_n\hat g_n$ at the $5\%$ level.}\label{tab:J}
\footnotesize\setlength{\tabcolsep}{4pt}\begin{tabular}{@{}lc ccc ccc c@{}}\toprule
& & \multicolumn{3}{c}{online $\mathcal J_n$} & \multicolumn{3}{c}{offline} & unproj.\\ \cmidrule(lr){3-5}\cmidrule(lr){6-8}\cmidrule(l){9-9}
& $n$ & $10\%$ & $5\%$ & $1\%$ & $10\%$ & $5\%$ & $1\%$ & $5\%$\\ \midrule
linear ($\chi^2_5$) & $10^4$ & 0.081 & 0.040 & 0.007 & 0.097 & 0.040 & 0.008 & 0.053 \\
 & $10^5$ & 0.103 & 0.043 & 0.004 & 0.096 & 0.044 & 0.003 & 0.053 \\
 & $10^6$ & 0.112 & 0.056 & 0.011 & 0.112 & 0.056 & 0.011 & 0.058 \\
\midrule
exponential ($\chi^2_3$) & $10^4$ & 0.118 & 0.050 & 0.015 & 0.118 & 0.065 & 0.014 & 0.089 \\
 & $10^5$ & 0.097 & 0.055 & 0.005 & 0.097 & 0.050 & 0.008 & 0.063 \\
 & $10^6$ & 0.094 & 0.048 & 0.010 & 0.100 & 0.045 & 0.011 & 0.050 \\
\midrule
power, $\vartheta=0.01$ & $10^4$ & 0.323 & 0.225 & 0.083 & 0.336 & 0.238 & 0.102 & 0.289 \\
 & $10^5$ & 0.982 & 0.955 & 0.870 & 0.982 & 0.955 & 0.873 & 0.964 \\
 & $10^6$ & 1.000 & 1.000 & 1.000 & 1.000 & 1.000 & 1.000 & 1.000 \\
\bottomrule\end{tabular}\end{table}

\subsection{Step size and units of measurement}\label{sec:sens}
Table~\ref{tab:sens} compares SGMM and the first-order version in the exponential model for steps $\gamma_i=c_\gamma n_0^{a-1}(n_0+i)^{-a}$ with $c_\gamma$ between $0.5$ and $8$, in the original units and after expressing $x_2$ (both as a regressor and as an instrument) in units ten times smaller, which multiplies the corresponding coefficient by $1/10$. In the original units, both methods work for $c_\gamma\ge1$, SGMM being somewhat more accurate for $\beta_1$ ($4.3$--$4.6$ against $5.0$--$6.8$) and both being less accurate with $c_\gamma=0.5$, where the steps are too small. After the change of units, SGMM gives exactly the same estimates of $\beta_0$ and $\beta_1$, because the Gauss--Newton step is invariant to the reparametrization (Remark~\ref{rem:design}(v)), whereas the first-order version diverges to the boundary of $\Theta$ for every $c_\gamma$. A first-order method needs a step size adapted to the scale of $\Phi\tr W\Phi$, which is unknown in advance; the Gauss--Newton preconditioner removes this dependence.

\begin{table}[t]\centering
\caption{Exponential IV model, $n=10^5$, $200$ replications: RMSE$\,\times10^3$ of the averaged estimate of $\beta_1$ (entries of diverging runs are reported without the factor $10^3$) for steps $\gamma_i=c_\gamma n_0^{a-1}(n_0+i)^{-a}$ and several $c_\gamma$, in the original units and after expressing $x_2$ (regressor and instrument) in units ten times smaller.}\label{tab:sens}
\small\begin{tabular}{@{}ll ccccc@{}}\toprule
& $c_\gamma$ & $0.5$ & $1$ & $2$ & $4$ & $8$\\ \midrule
original units & SGMM & 6.04 & 4.63 & 4.43 & 4.36 & 4.32 \\
 & first-order & 11.71 & 6.82 & 5.64 & 5.21 & 5.03 \\
$x_2$ rescaled & SGMM & 6.04 & 4.63 & 4.43 & 4.36 & 4.32 \\
 & first-order & 1.6 & 2.9 & 2.1 & 1.8 & 1.8 \\
\bottomrule\end{tabular}\end{table}

\subsection{Computing time}\label{sec:time}
Table~\ref{tab:time} compares a C implementation of SGMM for the exponential IV model with offline two-step GMM in NumPy. SGMM processes an observation in about $0.2$ microseconds, so one pass over $10^7$ observations takes $2.2$ seconds against $6.4$ seconds for offline GMM; offline GMM in addition needs the $763$ MB of data in memory, against less than $1$ KB of state for SGMM. The difference is larger when data arrive over time. If $10^6$ observations arrive in batches of $10^5$ and the offline estimate is recomputed after each batch (warm-started at the previous estimate), the total offline time is $2.54$ seconds, eleven times the $0.23$ seconds of SGMM (a single run; Table~\ref{tab:time} reports the best of three), for which an up-to-date estimate, variance and overidentification statistic are available after every observation. Wall-clock comparisons depend on the implementation; the point of Table~\ref{tab:time} is that the per-observation cost of SGMM is a few hundred floating-point operations and does not grow with $n$.

\begin{table}[t]\centering
\caption{Computing time (seconds, best of three) and memory in the exponential IV model: single-pass SGMM in C (one thread) versus offline two-step GMM in NumPy with all data in memory (multithreaded linear algebra), on the same two-core machine. SGMM stores $O(\dg^2+\dg\db)$ numbers (less than $1$ KB here).}\label{tab:time}
\small\begin{tabular}{@{}c cc c c@{}}\toprule
$n$ & SGMM & offline GMM & SGMM, $\mu$s per observation & data in memory (MB)\\ \midrule
$10^{5}$ & 0.018 & 0.036 & 0.18 & 8 \\
$10^{6}$ & 0.190 & 0.425 & 0.19 & 76 \\
$10^{7}$ & 2.172 & 6.392 & 0.22 & 763 \\
\bottomrule\end{tabular}\end{table}

\subsection{Stochastic volatility}\label{sec:sv}
Finally we apply SGMM to the stochastic volatility model of Example~\ref{ex:sv} with $(\omega,b,\sigma_u)=(-0.736,0.90,0.363)$, the design of \citet{AndersenSorensen1996}, which implies daily volatilities of about $2.5\%$. The data are geometrically mixing but not $L$-dependent, so this example is covered neither by Theorem~\ref{thm:fclt} nor exactly by Theorem~\ref{thm:dependent}. We run both methods in the parametrization $(\mu,b,\log\tau^2)$ of Example~\ref{ex:sv}, which improves the conditioning, use the moments of $100y_i$ (the parameters still refer to $y_i$, so that $\mu=\omega/(1-b)\approx-7.4$), a pilot sample of $n_0=5000$ (with a pilot estimate computed with weights from $\E[gg\tr]$), the default steps and the projection onto $\Theta=[-12,-3]\times[0.3,0.99]\times[-4,1]$, for $50$ replications. Table~\ref{tab:sv} compares (a) Algorithm~\ref{alg:sgmm}, whose weight matrix estimates $\E[gg\tr]^{-1}$, with offline two-step GMM using the same type of weight, and (b) SGMM with the Bartlett long-run weights of Section~\ref{sec:dependent} ($L=50$) with offline GMM using the Bartlett heteroskedasticity and autocorrelation consistent weight with the same bandwidth. With weights from $\E[gg\tr]$, SGMM is $1.3$ to $1.5$ times less accurate than offline GMM at $n=10^5$ and within $13\%$ at $n=10^6$ (RMSE of $\omega$: $0.0120$ against $0.0107$; ratios computed from unrounded values). With Bartlett long-run weights, both methods become more accurate at $n=10^6$, as the efficiency theory predicts, and SGMM is within $10\%$ of offline HAC GMM at $n=10^6$ (RMSE of $\omega$: $0.0099$ against $0.0091$). The projection was active in $43$ and $199$ steps in total in (a) and (b), respectively, over the $50$ replications (out of $5\times10^7$), all within the first $10^5$ observations. No replication failed. With larger early steps (a smaller offset $i_0$) in the original parametrization, the averaged estimates were noticeably biased and one replication was trapped near the boundary of $\Theta$; this was one of the reasons for adopting the default steps of Remark~\ref{rem:design}(v). Introducing a delay in the evaluation point, as in the proof of Theorem~\ref{thm:dependent}, had no visible effect with weights from $\E[gg\tr]$, and made the early iterations unstable with long-run weights. This is why the algorithm itself uses no delay.

\begin{table}[t]\centering
\caption{Stochastic volatility model with the $24$ moments of Example~\ref{ex:sv}, $(\omega,b,\sigma_u)=(-0.736,0.90,0.363)$, 50 replications, pilot $n_0=5000$, default steps; computations in the parametrization $(\mu,b,\log\tau^2)$. Bias and RMSE of SGMM and of offline two-step GMM with the same type of weight.}\label{tab:sv}
\footnotesize\setlength{\tabcolsep}{4pt}\begin{tabular}{@{}cl ccc ccc@{}}\toprule
& & \multicolumn{3}{c}{bias} & \multicolumn{3}{c}{RMSE}\\ \cmidrule(lr){3-5}\cmidrule(l){6-8}
$n$ & method & $\omega$ & $b$ & $\sigma_u$ & $\omega$ & $b$ & $\sigma_u$\\ \midrule
\multicolumn{8}{@{}l}{(a) weights from $\E[gg\tr]$}\\ \midrule
$10^4$ & offline GMM & 0.0067 & 0.0006 & -0.0137 & 0.1067 & 0.0146 & 0.0298 \\
 & SGMM & -0.0024 & -0.0006 & -0.0062 & 0.1502 & 0.0205 & 0.0414 \\
\cmidrule(l){2-8}
$10^5$ & offline GMM & -0.0077 & -0.0012 & -0.0007 & 0.0333 & 0.0046 & 0.0079 \\
 & SGMM & -0.0079 & -0.0012 & 0.0005 & 0.0441 & 0.0061 & 0.0116 \\
\cmidrule(l){2-8}
$10^6$ & offline GMM & -0.0018 & -0.0003 & 0.0003 & 0.0107 & 0.0015 & 0.0026 \\
 & SGMM & -0.0032 & -0.0005 & 0.0006 & 0.0120 & 0.0016 & 0.0030 \\
\midrule
\multicolumn{8}{@{}l}{(b) Bartlett long-run weights, $L=50$}\\ \midrule
$10^4$ & offline HAC GMM & -0.0879 & -0.0113 & -0.0013 & 0.1233 & 0.0163 & 0.0204 \\
 & SGMM, long-run weights & -0.0171 & -0.0023 & -0.0051 & 0.1510 & 0.0205 & 0.0433 \\
\cmidrule(l){2-8}
$10^5$ & offline HAC GMM & -0.0205 & -0.0027 & -0.0004 & 0.0359 & 0.0048 & 0.0075 \\
 & SGMM, long-run weights & -0.0155 & -0.0021 & 0.0010 & 0.0428 & 0.0058 & 0.0110 \\
\cmidrule(l){2-8}
$10^6$ & offline HAC GMM & -0.0013 & -0.0002 & 0.0000 & 0.0091 & 0.0012 & 0.0025 \\
 & SGMM, long-run weights & -0.0020 & -0.0003 & 0.0003 & 0.0099 & 0.0014 & 0.0027 \\
\bottomrule\end{tabular}\end{table}

%=====================================================================
\section{Empirical application: a third child and weeks worked}\label{sec:application}
%=====================================================================

\citet{AngristEvans1998} estimate the effect of having a third child on the labour supply of mothers, using the sex composition of the first two children as an instrument: parents of two boys or two girls are more likely to have a third child. We use their 1980 Census extract as distributed with the R package AER \cite{KleiberZeileis2008}: $n=254{,}654$ women aged $21$--$35$ with at least two children. The outcome is the number of weeks worked in 1979, $y$, which is zero for $47.2\%$ of the women and has mean $19.0$; $38.1\%$ of the women have a third child (\emph{morekids}).

Because $y$ is non-negative with many zeros, we use the multiplicative exponential model of Example~\ref{ex:nonlinear}, $y=\exp(x\tr\beta^*)\eta$ with $\E[\eta\mid z]=1$, where $x$ contains a constant, morekids, an indicator that the first child is a boy, age (centred at $30$ and divided by $5$) and its square, and indicators for Black, Hispanic and other race. The instruments $z$ are the exogenous variables together with separate indicators for two boys and for two girls, as in \citet{AngristEvans1998}. Thus $\db=8$, $\dg=9$, and there is one overidentifying restriction: the two instruments must imply the same effect of a third child. For comparison, linear 2SLS with the same instruments gives an effect of $-5.46$ weeks (standard error $1.23$), close to the estimates of \citet{AngristEvans1998}.

We compute offline two-step GMM on the full sample and SGMM in one pass, with a pilot sample of $n_0=5000$ and the default steps, for $20$ random orderings of the observations (the order matters for an online method; the data have no natural order). Table~\ref{tab:app} reports the results.

\begin{table}[t]\centering
\caption{Weeks worked and a third child: exponential IV model estimated by offline two-step GMM and by single-pass SGMM ($n=254{,}654$; pilot $n_0=5000$). Columns 4--5: SGMM for one random order of the observations, with plug-in standard errors; columns 6--7: mean and standard deviation of the SGMM estimates over $20$ random orders.}\label{tab:app}
\footnotesize\setlength{\tabcolsep}{4pt}\begin{tabular}{@{}l cc cc cc@{}}\toprule
& \multicolumn{2}{c}{offline GMM} & \multicolumn{2}{c}{SGMM, one order} & \multicolumn{2}{c}{SGMM, 20 orders}\\ \cmidrule(lr){2-3}\cmidrule(lr){4-5}\cmidrule(l){6-7}
& estimate & (s.e.) & estimate & (s.e.) & mean & s.d.\\ \midrule
constant & 2.9857 & (0.0241) & 2.9887 & (0.0231) & 2.9875 & 0.0061 \\
morekids & -0.2981 & (0.0703) & -0.2900 & (0.0729) & -0.2925 & 0.0168 \\
boy first & -0.0003 & (0.0048) & 0.0012 & (0.0048) & -0.0006 & 0.0016 \\
age & 0.2319 & (0.0084) & 0.2340 & (0.0084) & 0.2316 & 0.0025 \\
age$^2$ & -0.0181 & (0.0052) & -0.0164 & (0.0052) & -0.0172 & 0.0019 \\
black & 0.5150 & (0.0114) & 0.5143 & (0.0114) & 0.5147 & 0.0039 \\
hispanic & 0.0245 & (0.0143) & 0.0250 & (0.0145) & 0.0251 & 0.0044 \\
other race & 0.1061 & (0.0110) & 0.1050 & (0.0110) & 0.1057 & 0.0041 \\
\midrule $J$ statistic ($\chi^2_1$) & 4.91 & & 4.83 & & 4.68 & 0.66 \\
time (s) & 0.81 & & 0.20 & & & \\
\bottomrule\end{tabular}\end{table}

The offline estimate of the coefficient of morekids is $-0.298$ (standard error $0.070$), i.e.\ a third child reduces expected weeks worked by $26\%$ ($95\%$ interval $15\%$ to $35\%$), or about $4.9$ weeks at the mean, in line with the linear IV estimate. SGMM reproduces all coefficients and standard errors: averaged over the $20$ orders, the SGMM estimates differ from the offline ones by at most $0.16$ standard errors (by at most about one standard error in any single order), and the standard deviation across orders is between $0.24$ and $0.38$ standard errors, so the ordering of the data contributes a small fraction of the sampling uncertainty. The online overidentification statistic varies between $3.4$ and $6.4$ across orders (mean $4.7$), against $4.9$ for the offline statistic ($p=0.027$). The offline test, and the online test in $17$ of the $20$ orders, reject the hypothesis that the two instruments imply the same effect at the $5\%$ level (critical value $3.84$), which is consistent with the concern that boys and girls may affect household costs and labour supply differently \cite{RosenzweigWolpin2000}. One pass of SGMM took $0.16$ seconds on average over the orders ($0.20$ seconds in Table~\ref{tab:app}), against $0.81$ seconds for offline two-step GMM with the data in memory.

Two observations from this application led to choices made in the paper. First, the coefficient of morekids is weakly identified relative to the others: its standard error at $n_0=5000$ is about $0.5$, so single observations produce large Gauss--Newton steps. With larger early steps ($\gamma_0=1$, $i_0=100$), the final SGMM estimate was more than three standard errors away from the offline estimate in some coordinate in all $20$ orderings in exploratory runs, whereas the default steps of Remark~\ref{rem:design}(v) were stable in all of them. Second, in exploratory runs the unprojected statistic $n\hat g_n\tr W_n\hat g_n$ took values between $10$ and $36$. The averaged estimate is close to the offline estimate in every coordinate (within $0.32$ standard errors in the order of Table~\ref{tab:app}), but it is not exactly the minimiser of the GMM criterion, and the criterion is steep in a well-identified linear combination of the parameters. The projected statistic \eqref{eq:Jn} removes this effect.

%=====================================================================
\section{Conclusion}\label{sec:conclusion}
%=====================================================================

We have established the asymptotic theory of the nonlinear stochastic GMM algorithm: almost sure convergence of the iterates and of their average with a continuously updated weight matrix, almost sure rates, a functional central limit theorem with the efficient GMM variance under four moments, a consistent online variance estimator obtained as a by-product of the Gauss--Newton preconditioner, a projected online overidentification test, and an extension to $L$-dependent data with long-run variance weights. For IV-type moments, linear or not, each iteration costs $O(\dg^2+\dg\db+\db^2)$ operations thanks to a rank-two Woodbury update. In simulations with $1000$ replications, the averaged estimator was within a few percent of offline two-step GMM at $n=10^6$, the plug-in confidence intervals and the overidentification test behaved as predicted, the method was invariant to the units of measurement, and a compiled implementation was two to three times faster than offline GMM while using a fixed, small amount of memory. In the application to the labour supply effect of a third child, a single pass reproduced the offline estimates and standard errors, and the overidentification test led to the same conclusion in most orderings.

The results have limitations. Assumption~\ref{ass:mono} holds globally in linear models but only on a neighbourhood of $\beta^*$ in general, so that in nonlinear models the guarantees rely on a reasonable pilot estimate (Corollary~\ref{cor:datadriven}). The theory for dependent data assumes $L$-dependence with known $L$ and a stronger moment condition ($\kappa>\kappa_L(a)$, more than eight moments for $a=3/4$), and covers mixing data such as the stochastic volatility model only heuristically. The random-scaling intervals undercovered at our sample sizes with the default steps. Finally, the performance in finite samples depends on the choice of the pilot size and of the steps when some parameters are weakly identified, as the application shows; the default rule of Remark~\ref{rem:design}(v) worked in all our examples, but a data-driven choice would be useful. Natural next steps are a theory for mixing data with growing bandwidths, and online versions of other GMM-based tests, such as the endogeneity tests of \citet{ChenEtAl2025}, for nonlinear models.

%\appendix

\end{document}